\documentclass[runningheads,orivec]{llncs}
\usepackage[T1]{fontenc}

\usepackage{graphicx} %
\usepackage{enumitem}
\usepackage{amsmath,amssymb}
\usepackage{xcolor}
\usepackage{booktabs} %
\usepackage{tikz}
\usetikzlibrary{arrows.meta}
\usetikzlibrary{positioning}

\usepackage{multicol,multirow}

\usepackage{longtable}
\spnewtheorem{fact}{Fact}{\bfseries}{\itshape}
\spnewtheorem{observation}{Observation}{\bfseries}{\itshape}

\providecommand{\score}{\operatorname{score}}      %

\renewcommand{\setminus}{-}

\newcommand{\cala}{\ensuremath{\mathcal{A}}}
\newcommand{\cald}{\ensuremath{\mathcal{D}}}
\newcommand{\cale}{\ensuremath{\mathcal{E}}}
\newcommand{\cals}{\ensuremath{\mathcal{S}}}
\newcommand{\calt}{\ensuremath{\mathcal{T}}}
\newcommand{\dash}{\ensuremath{\mbox{-}}}
\newcommand{\card}[1]{\ensuremath{\lVert #1 \rVert}}
\newcommand{\caledt}{\cale\dash\calt}

\newcommand{\cc}{\ensuremath{\mathrm{CC}}}
\newcommand{\dc}{\ensuremath{\mathrm{DC}}}
\newcommand{\uw}{\ensuremath{\mathrm{UW}}}
\newcommand{\nuw}{\ensuremath{\mathrm{NUW}}}
\newcommand{\pv}{\ensuremath{\mathrm{PV}}}
\newcommand{\pc}{\ensuremath{\mathrm{PC}}}
\newcommand{\rpc}{\ensuremath{\mathrm{RPC}}}
\newcommand{\te}{\ensuremath{\mathrm{TE}}}
\newcommand{\tp}{\ensuremath{\mathrm{TP}}}
\newcommand{\ac}{\ensuremath{\mathrm{AC}}}
\newcommand{\dv}{\ensuremath{\mathrm{DV}}}
\newcommand{\av}{\ensuremath{\mathrm{AV}}}
\newcommand{\uac}{\ensuremath{\mathrm{UAC}}}

\newcommand{\nrv}{\ensuremath{{\rm NRV}}}

\title{%
Axiomatic
Tools for Separating Electoral Control Types, with
Applications to Concrete Systems%
}

\titlerunning{Axiomatic Tools for Separating Electoral Control Types}

\authorrunning{Chavrimootoo et al.}

\institute{%
Department of Computer Science, Denison University, Granville, OH 43023 \email{chavrimootoom@denison.edu}
\and 
Department of Computer Science, University of Rochester, Rochester, NY 14627%
\and 
Department of Computer Sciences, %
University of Wisconsin--Madison, Madison, WI 53706%
\and 
Khoury College of Computer Sciences, Northeastern University, Boston, MA 02115 %
\and 
Faculty of Mathematics and Physics, University of Ljubljana, Ljubljana, Slovenia%
\and 
Institute of Computer Science, Heinrich Heine University D\"{u}sseldorf, D\"{u}sseldorf, Germany%
\\[20pt]{\large June 10, 2026; revised August 24, 2026}%
}

\author{
    Michael C. Chavrimootoo\inst{1}%
    \and  
    Ian Clingerman\inst{2}%
    \and
    Ethan Ferland\inst{3}%
    \and
    Erin Gibson\inst{2}%
    \and 
    Lane A. Hemaspaandra\inst{2}%
    \and
    Quan Luu\inst{4}%
    \and
    David E. Narv\'{a}ez\inst{5}%
    \and
    Yanfei Wang\inst{6}%
}

\date{August 24, 2026}

\newcommand{\acktext}{This work was done in part while Chavrimootoo, Clingerman, Ferland, Gibson, Luu, and Narv\'{a}ez  were at the University of Rochester's Department of Computer Science, and while Narv\'{a}ez was also at Virginia Tech's Bradley Department of Electrical and Computer Engineering\@.  Clingerman is now located at M\&T Bank, 1 Seneca St., Buffalo, NY 14203, and Gibson is now located at Morgan Stanley, 750 7th Ave, New York, NY 10019.
Moreover, this work was supported in part
by NSF grants CCF-2006496 and 
DUE-2135431, CIFellows grant CIF2020-UR-36, and a Renewed Research Stay grant from the Alexander von~Humboldt Foundation.}

\date{\today} %

\begin{document} 
\sloppy

\maketitle

\begin{abstract}    
     Electoral control is the study of whether an attacker, by structural changes on an election such as adding/deleting/partitioning voters or candidates, can affect the winner in some desired way~\cite{bar-tov-tri:j:control}.
    Forty-four such attack types are often considered standard, and recently there has been work showing that sometimes the attack types---though seemingly distinct---in fact ``collapse,'' that is, for every input, either the attacker can achieve their goal under both of the control types or under neither of the control types.
    The papers doing this, however, while often exploiting axiomatic results that ensured collapses, found all the separations by human or computer-generated counterexamples. This left open the issue of whether even the separation direction can be driven by axiomatic results that allow large groups of separations to be almost automatically obtained. Our paper provides many such results, and we apply them to seven important voting systems, finding sixty-four new collapses and 1901 new separations. 
    We not only give axiomatic sufficient conditions and one complete characterization result, but also identify
    some control-problem pairs that universally separate---in other words, they separate under every voting rule. 
\end{abstract}

\section{Introduction}\label{sec:intro}

The four seminal papers of computational social choice by Bartholdi, Orlin, Tovey, and Trick~\cite{bar-orl:j:polsci:strategic-voting,bar-tov-tri:j:manipulating,bar-tov-tri:j:who-won,bar-tov-tri:j:control} cover the complexity of the winner problem and two types of attacks on elections: control and manipulation. 
Even now, more than three decades later, control and manipulation remain actively studied topics. This paper is on electoral control, 
which
refers to a class of attacks on elections where a chair has the ability to alter the structure of an election in order to change its outcome. 
The standard control approaches are adding/deleting/partitioning voters/candidates as defined by~\cite{bar-tov-tri:j:control} and subsequent papers (see~\cite{hem-hem-rot:j:destructive-control}).
As mentioned earlier, electoral control remains an active topic in computational social choice research to this day; e.g., see~\cite{%
alo-ina-jai-tal-mor:c:control-liquid-democracy,%
kac-rot:c:weighted-voting-games,%
mau-nic-nus-rot-see:c:completing-picture-schulze-rp,%
yan:c:impact-your-paper,%
alo-jan-lis-pap:c:perspective-liquid-democracy,%
che-kac-nus-rot-sch-see:c:control-in-comsoc,%
de-dey-san:c:control-in-polls,%
fal-jan-kno-pok-sch-slu-sor:c:project-strength-control,%
kac-rot:c:control-by-deleting-players-from-weighted-voting-games-is-np-pp-complete-for-the-penrose-banzhaf-power-index,%
che-gut-mus-sim:c:control-hedonic-games,%
sch-sor:c:control-participatory-budgeting,%
bui-cha-le-ngu:c:approximating-control,%
cha-jea:c:crc%
}.%

From the work above, 44 control types emerged, which are often referred to as the standard control types (although other control types have been defined, e.g., \cite{wan-yan:c:resolute-control-aamas-2017-anyone-but-them-yang-wang,gup-roy-sau-zeh:j:resolute-control,mau-nic-nus-rot-see:c:completing-picture-schulze-rp}).
For a long time, the standard control types were viewed as inherently different problems from each other as these problems were defined differently.
However, Hemaspaandra, Hemaspaandra, and Menton~\cite{hem-hem-men:j:search-versus-decision} questioned whether it was necessarily true that the standard control types are all pairwise different. In fact, they found that seven pairs of standard control types are exactly the same under any voting rule (we call this the ``general'' or ``universal'' case).
The natural takeaway here is that studying the complexity and properties of these problems separately essentially amounted to repeating the same work.

In this paper, we use the framework introduced by Carleton et al.~\cite{car-cha-hem-nar-tal-wel:j:sct} to study when pairs of control types equal each other (i.e., ``collapse'' or ``form collapsing pairs''); in their framework, they consider the $44$ standard control types and establish that under approval and veto, there were 15 previously unknown collapses, and they prove that under plurality, approval, and veto the only collapsing pairs of control types are the ones identified in their own work and in~\cite{hem-hem-men:j:search-versus-decision}.
The work of Carleton et al.~\cite{car-cha-hem-nar-tal-wel:j:sct} shows that it is worthwhile to study control-type relationships both in the universal case and for concrete voting rules: doing so helps avoid duplicate work and asks deeper questions about how control types interrelate.
Footnote~1 of \cite{car-cha-hem-nar-tal-wel:j:sct} gives examples of previous work that was duplicated in this way.
In a follow-up paper~\cite{car-cha-hem-nar-tal-wel:j:search-versus-search}, they also devised a framework to study 
separations and collapses for electoral control types in the \emph{search} model.
We provide a new type of result that was not previously considered: universal separations (Theorem~\ref{t:universal-separations}). Prior work has considered the notion of universal collapses, and the fact that certain pairs of control types always separate (regardless of the voting rules) is a new type of result! 
As our Theorem~\ref{t:characterize-uw-nuw}, we completely characterize the condition for certain pairs of control types to collapse/separate.
Additionally, we provide many axiomatic-sufficient conditions to obtain separation results (Theorems~\ref{t:all-cc-dc} and \ref{t:sufficient-group4}--\ref{t:sufficient-group4bii}).
Part of the appeal of studying universal separations and axiomatic-sufficient conditions for separations is that it simplifies greatly how we conduct our analyses and make this type of work more intuitive and less ad hoc. Indeed, all the of separations established by Carleton et al.~\cite{car-cha-hem-nar-tal-wel:j:sct} were done on a case-by-case basis, at times using computer-aided searches.
Our paper itself is an example of the power of axiomatic-based separations and collapses as the majority of the separations we find follow from Theorems~\ref{t:universal-separations}--\ref{t:sufficient-group4bii}.

We also pursue an explicitly stated open question posed by Carleton et al.~\cite{car-cha-hem-nar-tal-wel:j:sct} by studying several important voting rules that have not yet been treated systematically from this separations/collapses perspective, namely $k$-NRV, Llull, Copeland, Schulze, Ranked Pairs, Fallback, and Bucklin, which have been carefully studied in the computational social choice literature~\cite{men:j:range-voting,fal-hem-hem-rot:j:llull,men-sin:c:schulze,hem-lav-men:j:schulze-and-ranked-pairs,erd-rot:c:fallback,erd-fel-rot-sch:j:control-in-bucklin-and-fallback-voting}. 

\section{Preliminaries}\label{sec:preliminaries}

For consistency, we follow the standard definitions and notation used in the work closest to ours, namely that of Hemaspaandra, Hemaspaandra, and Menton~\cite{hem-hem-men:j:search-versus-decision} and Carleton et al.~\cite{car-cha-hem-nar-tal-wel:j:sct}.

\subsection{Elections and Electoral Control}
An election is a pair $(C, V)$, where $C$ is a finite set of candidates and $V$ is a collection of votes over $C$. 
Unless specified otherwise, each vote is a complete, transitive, and tie-free ranking of all the candidates in $C$. 
For example, if $C = \{a,b,c\}$, one such vote is $b > a > c$, where $b$ is ranked first and $c$ is ranked last.
A voting rule $\cale$ is a function that maps an election $(C, V)$ to a (possibly empty) subset of $C$, i.e., the winner set, which we denote using the notation $\cale(C, V)$. 
As is standard, whenever we speak of an election $(C', V)$ whose votes are over a set $C \supseteq C'$, we implicitly mean that the votes have been restricted to be over $C'$ (i.e., the votes are masked down).

(Electoral) control is a class of attacks on elections where an entity that oversees the election has the ability to alter the structure of the election to yield a preferred outcome. 
The standard control approaches are adding/deleting/partitioning candidates/voters, and those will be the control approaches this paper will focus on. Definition~\ref{def:control} provides concrete definitions for these problems.
We
want to focus 
this
presentation on our condition-based results rather than our computer-aided searches, so we present in the main body only our results that are about adding/deleting candidates/voters, and we defer all results about partitioning to the appendix. 
By the same token, we also defer results that are about concrete voting rules, along with the definitions for those voting rules, to the appendix. We also summarize those results in Table~\ref{tab:concrete-summary} for the convenience of the reader.

Each problem in Definition~\ref{def:control} is about constructive control, where the goal is to make a designated candidate a winner. This is denoted in the problem name's abbreviation by using ``-CC'' after the voting rule's name.
Moreover, each such problem is in the nonunique-winner model (denoted using the ``-NUW'' suffix), where a winner does not need to be a unique winner. 

\begin{definition}[\cite{hem-hem-men:j:search-versus-decision,car-cha-hem-nar-tal-wel:j:sct}]\label{def:control}  
Let $\cale$ be a voting rule.
\begin{enumerate}
    \item\label{ccac} In the \textbf{constructive control by adding candidates} problem for $\cale$
    (denoted by $\cale\dash\cc\dash\ac\dash\nuw$), we are given two
    disjoint sets of candidates $C$ and $A$, $V$ a collection of votes over $C \cup A$, a
    candidate $p \in C$, and a nonnegative integer $k$. We ask if there is a set
    $A' \subseteq A$ such that (i)~$\|A'\| \leq k$ and (ii)~$p$ is a winner of $\cale$ election
    $(C\cup A', V)$.
    
    \item\label{ccuac} In the \textbf{constructive control by unlimited adding candidates} problem for $\cale$ (denoted by $\cale\dash\cc\dash\uac\dash\nuw$),\footnote{Definition~\ref{def:control}.\ref{ccuac} was the definition used in the seminal work of Bartholdi, Tovey, and Trick~\cite{bar-tov-tri:j:control} for the constructive control by adding candidates problem. However, this was asymmetric with the definitions in that paper and later work restated the problem in the way captured in Definition~\ref{def:control}.\ref{ccac}. Because constructive control by unlimited adding candidates was actively studied in COMSOC, and because we wish to keep our study close to that of Carleton et al.~\cite{car-cha-hem-nar-tal-wel:j:sct}, we also study this problem. For more details and discussions on the topic, we refer interested readers to \cite[Footnote~5]{car-cha-hem-nar-tal-wel:j:sct}.} we are given two
    disjoint sets of candidates $C$ and $A$, $V$ a collection of votes over $C \cup A$, and a
    candidate $p \in C$. We ask if there is a set
    $A' \subseteq A$ such that $p$ is a winner of $\cale$ election $(C\cup A', V)$.

    \item\label{ccdc} In the \textbf{constructive control by deleting candidates} problem for $\cale$ (denoted by $\cale\dash\cc\dash\dc\dash\nuw$), we are given an election
    $(C, V)$, a candidate $p \in C$, and a nonnegative integer $k$. We ask if there is a set
    $C' \subseteq C$ such that (i)~$\|C'\| \leq k$, (ii)~$p \not\in C'$, and
    (iii)~$p$ is a winner of $\cale$ election
    $(C - C', V)$.
    
    \item\label{ccav} In the \textbf{constructive control by adding voters} problem for $\cale$
    (denoted by $\cale\dash\cc\dash\av\dash\nuw$), we are given a set of candidates $C$,
    two collections of votes, $V$ and $W$, over $C$, a candidate $p \in C$, and a nonnegative
    integer $k$. We ask if there is a 
    collection $W' \subseteq W$ such that (i)~$\|W'\| \leq k$ and (ii)~$p$ is a winner of $\cale$ election $(C, V \cup W')$.
    
    \item\label{ccdv} In the \textbf{constructive control by deleting voters} problem for $\cale$
    (denoted by $\cale\dash\cc\dash\dv\dash\nuw$), we are given an election
    $(C, V)$, a candidate $p \in C$, and a nonnegative integer $k$. We ask if there is a collection
    $V' \subseteq V$ such that (i)~$\|V'\| \leq k$ and (ii)~$p$ is a winner of $\cale$ election
    $(C, V-V')$.

    \end{enumerate}
\end{definition}

For a given voting rule $\cale$, Definition~\ref{def:control} gives five control problems---aka, \emph{control types}---that are all about constructive control and in the nonunique-winner model. Each such type can be analogously defined to be about destructive control---wherein the goal is to prevent a candidate from winning---by changing each instance of ``is a winner'' to ``is not a winner'' in Definition~\ref{def:control}. For each of the five types above, the corresponding types name's abbreviation (in the destructive setting) uses ``-DC'' after the voting rule's name (instead of ``-CC''). Doing so then defines 10 control types.
The work on destructive control was initiated by Hemaspaandra, Hemaspaandra, and Rothe~\cite{hem-hem-rot:j:destructive-control}.
Finally, each of those 10 types can be analogously defined in the unique-winner model by changing each instance of ``a winner'' to ``a unique winner'' in the definitions of the 10 types.
Both winner models have been considered in the literature, so we consider them both in our work too.
The corresponding abbreviation of each type's name (in the unique-winner model) is suffixed with ``-UW'' instead of ``-NUW'' thus defining 20 of the control types. The remaining 24 standard control types are about partitioning and are defined in the appendix.

In this paper---as has been done within the literature on electoral control---we will use ``(control) action'' to refer to a concrete alteration to an election, e.g., with respect to $\cale$-CC-AC-UW, such an action refers to a selection of candidates to be added to the election. 
We will also use ``partial type'' to a restriction of a control type without the voting rule, e.g., CC-AC-UW is an example of a partial type.
To improve the readability of our paper, we at times omit the voting rule in a control type, and it will always be clear from context whether we are referring to a control type or a partial control type (so we may write CC-AC-UW instead of $\cale$-CC-AC-UW, which is particularly helpful when the voting rule has a long name).

\subsection{Framework for Separations and Collapses}

Each of the above-defined control types are decision problems, and so they are naturally associated with a set/language of Yes instances. To compare two control types, e.g., $\cale\dash\cc\dash\av\dash\uw$ and $\cale\dash\cc\dash\dv\dash\uw$, we compare the two types as \emph{sets}. However, we only do so in the case where two control types are compatible, i.e., they have the same input fields. For example, an instance of $\cale\dash\cc\dash\av\dash\uw$ is a tuple $(C, V, W, p, k)$ while an instance of $\cale\dash\cc\dash\dv\dash\uw$ is a tuple $(C, V, p, k)$. If two control types have different input types, they are incompatible. Otherwise, they are said to be compatible. Of the $\binom{44}{2} = 946$ pairs of control types, only $3\cdot\binom{4}{2} + \binom{8}{2} + \binom{24}{2} = 322$ are compatible~\cite{car-cha-hem-nar-tal-wel:j:sct}, and Table~\ref{tab:compatible} provides a list of control types (about the same voting rule) that are compatible.

\begin{table}[htb]
    \centering
    \caption{Compatibility groups for control types.}\label{tab:compatible}
    \begin{tabular}{|c|c|c|}
        \hline
        Group Number & Control Problems in Group\\\hline
        1 & Control by unlimited adding of candidates\\
        2 & Control by (limited) adding of candidates\\
        3 & Control by (limited) adding of voters\\
        4 & Control by deleting candidates/voters\\
        5 & Control by partitioning of voters and (run-off) partitioning of candidates\\
        \hline
    \end{tabular}
\end{table}

Two (compatible) control types $\calt_1$ and $\calt_2$ are said to collapse if $\calt_1 = \calt_2$. Otherwise, they are said to separate. 
Understanding how control types separate is also of interest. So if $\calt_1$ and $\calt_2$ separate, we also investigate if $\calt_1 \subsetneq \calt_2$ or $\calt_1 \supsetneq \calt_2$, or if neither of the previous two relationships hold. If that ``neither'' case holds, we say that $\calt_1$ and $\calt_2$ are incomparable.\footnote{Carleton et al.~\cite{car-cha-hem-nar-tal-wel:j:sct} provide a richer discussion of the notion of incomparability and give two additional notions of incomparability, namely weak incomparability and strong incomparability. In this paper, we do not concern ourselves with those notions, but we mention in passing that some of our results are indeed about strong incomparability and none are about weak incomparability.} 
Thus this paper is concerned with pairs of compatible control types that are about the same voting rule and their relationships as described in this paragraph.

Given a pair of partial control types $(\cala_1, \cala_2)$, if for each voting rule $\cale$ the relationship between $\cale\dash\cala_1$ and $\cale\dash\cala_2$ is always the same (e.g., they always collapse or always separate), then we say that the relationship holds universally (or ``in the general case''). 
As a consequence, if a relationship does not hold under a given voting rule, it certainly does not hold universally. 
Hemaspaandra, Hemaspaandra, and Menton~\cite{hem-hem-men:j:search-versus-decision} and Carleton et al.~\cite{car-cha-hem-nar-tal-wel:j:sct} provide the only known universal results in the literature, and this paper will provide new ones in Section~\ref{sec:nonpartition}. 

For each collapse or containment result, we give a proof. Our separations were computer-generated using our own programs. We list in the appendix the separations and the instances that witness those separations. Each separation witness we give was independently verified by a separately written program. Our code can be found at the following URL: \url{https://github.com/MikeChav/RTR\_Code\_Public}.

\section{Universal Separations and Axiomatic Results}\label{sec:nonpartition}

Following our motivation from the Introduction, we here present three types of results. First, we give universal separations (Theorem~\ref{t:universal-separations}), i.e., separations that hold under every voting rule, which is a notion 
that is analogous
to that of universal collapses established by~\cite{hem-hem-men:j:search-versus-decision}. Second, we give axiomatic results (Theorems~\ref{t:all-cc-dc} and \ref{t:sufficient-group4}--\ref{t:sufficient-group4bii}), i.e., separations/collapses that hold under voting rules satisfying certain relationships. Finally, we show that whether $\cale$ is weakly resolute completely characterizes the collapse (or noncollapse) of 10 pairs of control types (Theorem~\ref{t:characterize-uw-nuw}).

\begin{theorem}\label{t:universal-separations}
    Let $\cale$ be a voting rule, and let $\caledt_1$ and $\caledt_2$ be two control types.
    If $\caledt_1$ is a constructive control type and $\caledt_2$ is a destructive control type, 
    then $\cale\dash\calt_1 \neq \cale\dash\calt_2$.
\end{theorem}
\begin{proof}
Let $\caledt_1$ and $\caledt_2$ be two control types that satisfy the requirements of the theorem.
Let $(C, V)$ be a one-candidate and one-vote $\cale$ election, and let $C = \{p\}$.
Because $\|C\| = 1$, a winner of $\cale$ election $(C, V)$ must be a unique winner.
Suppose for the sake of contradiction that $\caledt_1 = \caledt_2$. We will prove our theorem by considering which group the control types are a part of.

  \textbf{Group~1 case.}
  Consider the $\caledt_1$ instance $I = (C, V, p, \emptyset)$, where the set of spoiler candidates is empty. Because the spoiler candidate set is empty, there is only one possible control action under $\uac$ (adding no candidates), and the outcome of the election under $\uac$ is the same as the outcome of the original $\cale$ election $(C, V)$. 
  
  By our assumption, $I \in \caledt_1 \iff I \in \caledt_2$. If $\{p\} = \cale(C, V)$, then $I\in \caledt_1$, and there is no control action to prevent $p$ from being a winner, so $I\not\in\caledt_2$. If $\{p\} \neq \cale(C, V)$, then $I\not\in\caledt_1$, and there is no control action to make $p$ a winner, so $I\in\caledt_2$. In both cases, we have a contradiction.

  \textbf{Group~2 case.}
  Consider the $\caledt_1$ instance $I = (C, V, p, \emptyset, 0)$, where the set of spoiler candidates is empty and the limit $k$ is 0. Because the spoiler candidate set is empty, there is only one possible control action under $\ac$ (adding no candidates), and the outcome of the election under $\ac$ is the same as the outcome of the original $\cale$ election $(C, V)$. 
  
  By our assumption, $I \in \caledt_1 \iff I \in \caledt_2$. If $\{p\} = \cale(C, V)$, then $I\in \caledt_1$, and there is no control action to prevent $p$ from being a winner, so $I\not\in\caledt_2$. If $\{p\} \neq \cale(C, V)$, then $I\not\in\caledt_1$, and there is no control action to make $p$ a winner, so $I\in\caledt_2$. In both cases, we have a contradiction.

  \textbf{Group~3 case.}
  Consider the $\caledt_1$ instance $I = (C, V, p, \emptyset, 0)$, where the set of spoiler votes is empty and the limit $k$ is 0. Because the spoiler votes set is empty, there is only one possible control action under $\av$ (adding no votes), and the outcome of the election under $\av$ is the same as the outcome of the original $\cale$ election $(C, V)$. 
  
  By our assumption, $I \in \caledt_1 \iff I \in \caledt_2$. If $\{p\} = \cale(C, V)$, then $I\in \caledt_1$, and there is no control action to prevent $p$ from being a winner, so $I\not\in\caledt_2$. If $\{p\} \neq \cale(C, V)$, then $I\not\in\caledt_1$, and there is no control action to make $p$ a winner, so $I\in\caledt_2$. In both cases, we have a contradiction.

  \textbf{Group~4 case.}
  Consider the $\caledt_1$ instance $I = (C, V, p, 0)$, where the  limit $k$ is 0. Because the limit is 0, there is only one possible control action under $\dc/\dv$ (deleting no candidate/voter, depending on what $\caledt_1$ is about), and the outcome of the election under $\dv$/$\dc$ is the same as the outcome of the original $\cale$ election $(C, V)$. 
  
  By our assumption, $I \in \caledt_1 \iff I \in \caledt_2$. If $\{p\} = \cale(C, V)$, then $I\in \caledt_1$, and there is no control action to prevent $p$ from being a winner, so $I\not\in\caledt_2$. If $\{p\} \neq \cale(C, V)$, then $I\not\in\caledt_1$, and there is no control action to make $p$ a winner, so $I\in\caledt_2$. In both cases, we have a contradiction.

  \textbf{Group~5 case.}
  Consider the $\caledt_1$ instance $I = (C, V, p)$. Because $\|C\|=1$, any subelection can have at most one winner, so control using the $\te$ rule is the same as using the $\tp$ rule, so we will only consider if $\caledt_1$ is about partitioning candidates or voters.
  
  By our assumption, $I \in \caledt_1 \iff I \in \caledt_2$. We consider two cases. 
  \begin{description}[wide]
    \item[Case~1:] If $\{p\} = \cale(C, V)$. Let $\calt' \in \{\calt_1, \calt_2\}$.
    If $\calt'$ is about partitioning candidates, then under both possible partitions---that is, $(\emptyset, C)$ and $(C, \emptyset)$---the final stage is $(C, V)$. Similarly, if $\calt'$ is about partitioning voters, then under both possible partitions---that is, $(\emptyset, V)$ and $(V, \emptyset)$---the final stage is $(C, V)$. So $p$ is always a unique winner under each $\calt'$ control action. We can thus conclude that $I \in \caledt_1$ and $I\not\in \caledt_2$, which is a contradiction.

    \item[Case~2:] If $\{p\} \neq \cale(C, V)$. Let $\calt' \in \{\calt_1, \calt_2\}$.
    If $\calt'$ is about partitioning candidates, then under both possible partitions---that is, $(\emptyset, C)$ and $(C, \emptyset)$---the final stage is $(\emptyset, V)$. Similarly, if $\calt'$ is about partitioning voters, then under both possible partitions---that is, $(\emptyset, V)$ and $(V, \emptyset)$---the final stage is $(\emptyset, V)$. So $p$ is never a winner under each $\calt'$ control action. We can thus conclude that $I \not\in \caledt_1$ and $I\in \caledt_2$, which is a contradiction.
  \end{description}

We have shown that each case leads to a contradiction of our initial assumption, which concludes the proof of this theorem.\qed
\end{proof}

This universal separation is, to our knowledge, the first of its kind and 
provides a new insight:
seeking a collapse between a constructive control type and a destructive control type is hopeless. However, recall that we are not solely interested in determining which pairs of control types collapse or separate, but we are also interested in how they separate, by proving strict containments or incomparabilty results. The above result unfortunately does not provide us with that type of insight. In particular, although we obtain $\caledt_1 \neq \caledt_2$, the proof does not establish if the ``$\neq$'' is in fact a ``$\not\subseteq$'' or a ``$\not\supseteq$'' or both.
Nonetheless, we are able to prove a corresponding incomparability result (Theorem~\ref{t:all-cc-dc}) under a natural and easy-to-satisfy (at least for the voting rules in this paper) condition by modifying the proof of Theorem~\ref{t:universal-separations}. We say ``natural'' because any voting rule that fails to satisfy the condition in Theorem~\ref{t:all-cc-dc} would be one that does not allow a voter to express a strict preference in 2-candidate and 1-candidate elections.

\begin{theorem}\label{t:all-cc-dc}
Let $\cale$ be a voting rule such that there is a vote $v$ over candidates $\{p, q\}$ for which $\cale(\{p,q\}, \{v\}) = 
  \cale(\{p\}, \{v\}) = \{p\}$.
If $\caledt_1$ and $\caledt_2$ are standard, compatible control types such that $\caledt_1$ is a constructive control type and $\caledt_2$ is a destructive control type, it follows that $\cale\dash\calt_1$ and  $\cale\dash\calt_2$ are incomparable.
\end{theorem}
\begin{proof}
Let $\cale$ be a voting rule such that there is a vote $v$ over candidates $\{p,q\}$ and $\cale(\{p,q\}, \{v\}) = \cale(\{p\}, \{v\}) = \{p\}$. Moreover, let $\caledt_1$ and $\caledt_2$ be two control types as specified in the theorem statement. 
We will prove our theorem by considering which group the control types are a part of. Let $C = \{p,q\}$ and let $V = \{v\}$.

  \textbf{Group~1 case.}
  Let the spoiler set of candidates be empty.
  Then there is only one possible control action under $\uac$ (adding no candidates), and the outcome of the election under $\uac$ is the same as the outcome of the original $\cale$ election $(C, V)$. 
  Therefore, $p$ is always a unique winner, i.e., $(C, V, p, \emptyset) \in \caledt_1 - \caledt_2$, and $q$ is never a winner, i.e., $(C, V, q, \emptyset) \in \caledt_2 - \caledt_1$.

  \textbf{Group~2 case.}
  Let the set of spoiler candidates be empty and the limit $k$ be 0.
  Then there is only one possible control action under $\ac$ (adding no candidates), and the outcome of the election under $\ac$ is the same as the outcome of the original $\cale$ election $(C, V)$.
  Therefore, $p$ is always a unique winner, i.e., $(C, V, p, \emptyset, 0) \in \caledt_1 - \caledt_2$, and $q$ is never a winner, i.e., $(C, V, q, \emptyset, 0) \in \caledt_2 - \caledt_1$.

  \textbf{Group~3 case.}
  Let the set of spoiler votes be empty and the limit $k$ be 0.
  Then there is only one possible control action under $\av$ (adding no votes), and the outcome of the election under $\av$ is the same as the outcome of the original $\cale$ election $(C, V)$. 
  Therefore, $p$ is always a unique winner, i.e., $(C, V, p, \emptyset, 0) \in \caledt_1 - \caledt_2$, and $q$ is never a winner, i.e., $(C, V, q, \emptyset, 0) \in \caledt_2 - \caledt_1$.

  \textbf{Group~4 case.}
  Let the deletion limit $k$ be 0.
  Then there is only one possible control action under $\dc/\dv$ (deleting no candidate/voter, depending on what $\caledt_1$ and $\caledt_2$ are), and the outcome of the election under $\dv$/$\dc$ is the same as the outcome of the original $\cale$ election $(C, V)$. 
  Therefore, $p$ is always a unique winner, i.e., $(C, V, p, 0) \in \caledt_1 - \caledt_2$, and $q$ is never a winner, i.e., $(C, V, q, 0) \in \caledt_2 - \caledt_1$.

  \textbf{Group~5 case.}
  By the assumptions of the theorem, $p$ uniquely wins any subelection it is present in as long as the collection of voters is $V$. If partitioning candidates, the collection of voters never changes, and so $p$ always proceeds to the final round, which $p$ wins uniquely. Moreover, $q$ never wins if $p$ is present, so $q$ can never be a final round winner. 
  If partitioning voters, the only two partitions possible are $(\emptyset, \{v\})$ and $(\{v\}, \emptyset)$. By the symmetry of \pv, both yield the same outcomes, so let us only consider one of the two. Notice that $p$ uniquely wins subelection $(\{p,q\}, \{v\})$, so $p$ always proceeds to the final round, which it wins uniquely. Moreover, $q$ never wins if $p$ is present, so $q$ can never be a final-round winner. 
  In both cases, $(C, V, p) \in \caledt_1 - \caledt_2$ and $(C, V, q) \in \caledt_2 - \caledt_1$.~\qed
\end{proof}

In terms of our search for results, this is rather significant as it takes care of the comparison of (and separation witnesses for) $2^2 + 2^2 + 2^2 + 4^2 + 12^2 = 172$ compatible pairs out of 322 compatible pairs! This is also quite helpful when considering the control types from Groups~1--3, as the above result covers half of them. We now give another result that characterizes the remaining relationships between control types from Groups~1--3.

A voting rule $\cale$ is said to be weakly resolute if for each election $(C, V)$, $\card{\cale(C, V)} \leq 1$, i.e., $\cale$ never selects more than one winner. We use weak resoluteness to completely characterize the conditions for 10 different pairs of control types to collapse. 
The prior work by Hemaspaandra, Hemaspaandra, and Menton~\cite{hem-hem-men:j:search-versus-decision} and by Carleton et al.~\cite{car-cha-hem-nar-tal-wel:j:sct} did not present these collapsing pairs of control types.

\begin{theorem}\label{t:characterize-uw-nuw}
  Let $\cale$ be a voting rule and let $\calt \in \{\uac,\ac,\av,\dc,\dv\}$. The following are equivalent.
  \begin{enumerate}
      \item\label{i:weakly-resolute-assumption} $\cale$ is weakly resolute.
      \item\label{i:weakly-resolute-cc} $\cale\dash\cc\dash\calt\dash\uw = \cale\dash\cc\dash\calt\dash\nuw$.
      \item\label{i:weakly-resolute-dc} $\cale\dash\dc\dash\calt\dash\nuw = \cale\dash\dc\dash\calt\dash\uw$.
  \end{enumerate}
\end{theorem}
\begin{proof}
  Let $\cale$ be a voting rule and let $\calt \in \{\uac,\ac,\av,\dc,\dv\}$.

  \ref{i:weakly-resolute-assumption} $\implies$ \ref{i:weakly-resolute-cc}: Assume that $\cale$ is weakly resolute.  
  By definition, $\cale\dash\cc\dash\calt\dash\uw \subseteq \cale\dash\cc\dash\calt\dash\nuw$. We now prove that $\cale\dash\cc\dash\calt\dash\nuw \subseteq \cale\dash\cc\dash\calt\dash\uw$
  Fix $I \in \cale\dash\cc\dash\calt\dash\nuw$ (if $I$ does not exist, then $\cale\dash\cc\dash\calt\dash\nuw = \emptyset$ and as a consequence $\cale\dash\cc\dash\calt\dash\uw = \emptyset$, giving the equality we seek). Since $\cale$ is weakly resolute, any candidate that can be made a winner under $\calt$, is also made a unique winner. So $I \in \cale\dash\cc\dash\calt\dash\uw$, which proves that $\cale\dash\cc\dash\calt\dash\nuw \subseteq \cale\dash\cc\dash\calt\dash\uw$.

  \ref{i:weakly-resolute-cc} $\implies$ \ref{i:weakly-resolute-dc}:
  Suppose that $\cale\dash\cc\dash\calt\dash\uw = \cale\dash\cc\dash\calt\dash\nuw$. By definition, $\cale\dash\dc\dash\calt\dash\nuw \subseteq \cale\dash\dc\dash\calt\dash\uw$, so it suffices to prove that $\cale\dash\dc\dash\calt\dash\uw \subseteq \cale\dash\dc\dash\calt\dash\nuw$. Fix an arbitrary $I \in \cale\dash\dc\dash\calt\dash\uw$ (if $I$ does not exist, then $\cale\dash\dc\dash\calt\dash\uw = \emptyset$ and as a consequence $\cale\dash\dc\dash\calt\dash\nuw = \emptyset$, giving the equality we seek). Let $(C, V)$ be the specified election in $I$ and $p$ be the distinguished candidate specified in $I$. By the assumption on $I$, there is a control action (the specific action being tied to what $\calt$ is) such that $p$ can be prevented from being a unique winner. Thus, under that control action, $p$ is either a tied winner with some other candidate $d$ or $p$ is not a winner at all. If the latter case holds, then clearly $I \in \cale\dash\dc\dash\calt\dash\nuw$. Now suppose that the former case holds and let $I'$ denote the input $I$ with the distinguished candidate replaced with $d$. In our current case, there is a control action under which $d$ is a winner (i.e., $I' \in \cale\dash\cc\dash\calt\dash\nuw$), and so we can use our assumption to conclude that there is a control action under which $d$ is a \emph{unique} winner (because $I' \in \cale\dash\cc\dash\calt\dash\uw$). Using that same control action that makes $d$ is a unique winner, we have a control action that prevents $p$ from being a winner. Since $I$ and $I'$ only differ in the specified distinguished candidate, it follows that $I \in \cale\dash\dc\dash\calt\dash\nuw$.

  \ref{i:weakly-resolute-dc} $\implies$ \ref{i:weakly-resolute-assumption}:
  Suppose that $\cale\dash\dc\dash\calt\dash\nuw = \cale\dash\dc\dash\calt\dash\uw$, and suppose for the sake of contradiction that $\cale$ is not weakly resolute. Then there is an election $(C, V)$ such that $\|\cale(C, V)\| > 1$. Fix $p \in \cale(C,V)$. Furthermore, define 

  $$I = \begin{cases}
    (C, V, p, \emptyset) &\text{ if } \calt = \uac,\\
    (C, V, p, \emptyset, 0) &\text{ if } \calt \in \{\ac,\av\}\text{, and}\\
    (C, V, p, 0) &\text{ if } \calt \in \{\dc,\dv\}.
  \end{cases}$$
  
  It follows that $I \in \cale\dash\dc\dash\calt\dash\uw$ since $p$ is not a unique winner of $\cale$ election $(C,V)$ and the winner set under any control action is $\cale(C,V)$ (because in each case, the only available control action does not change $C$ or $V$). But that in turn implies that there is no control action under which $p$ can be prevented from being a winner, so $I\not\in\cale\dash\dc\dash\calt\dash\nuw$. This contradicts our initial assumption. So it must be that $\cale$ is weakly resolute.\qed
\end{proof}

By Theorems~\ref{t:all-cc-dc} and~\ref{t:characterize-uw-nuw}, we have now completely determined the relationships between control types from Groups~1--3, under the voting rules in this paper and any other natural voting rules. 
This result is quite useful as it characterizes collapses between two nonpartition control types that differ only in being about the \uw\ model versus the \nuw\ model. And so, for a voting rule that is not weakly resolute, we get that every such nonpartition compatible pair separates, but we don't know that this certainly holds for the partition types. For example, 
we know approval voting and veto are not weakly resolute rules, and yet it holds that relative to each of these voting rules, $\dc\dash\pv\dash\te\dash\uw = \dc\dash\pv\dash\te\dash\nuw$ (defined in Appendix~\ref{app:defs}). Therefore, proving corresponding characterizations for partition-based control types seems to require a different type of argument. 
Nonetheless, we mention without proof---as it is easy to see---that weak-resoluteness still provides an axiomatic-sufficient for various collapses, and such an example is given as part of Appendix~\ref{subsec:ranked-pairs}.

We now turn our attention to control types from Group~4. The remaining theorems in this section shed light on the relationships between the remaining pairs in Group~4. (Note that the separations in Theorem~\ref{t:sufficient-group4} overlap with those in Theorem~\ref{t:characterize-uw-nuw} exactly in the case where $\cald_1 = \cald_2$.) %
\begin{theorem}\label{t:sufficient-group4}
    Let $\cale$ be a voting rule that is not weakly resolute. Then, for each $\cald_1, \cald_2 \in \{\dc, \dv\}$,
    it holds that
        $\cale\dash\cc\dash\cald_1\dash\nuw \not\subseteq \cale\dash\cc\dash\cald_2\dash\uw$ and
        $\cale\dash\dc\dash\cald_1\dash\uw \not\subseteq \cale\dash\dc\dash\cald_2\dash\nuw$.
\end{theorem}
\begin{proof}
    Fix $\cale$, $\cald_1$, and $\cald_2$ to satisfy the theorem statement's requirements.
    Since $\cale$ is not weakly resolute, there is an election $(C, V)$ with at least two winners. Let $p$ denote one of those winners, and let $I=(C, V, p, 0)$. Clearly, $I$ is in both $\cale\dash\cc\dash\cald_1\dash\nuw$ and $\cale\dash\dc\dash\cald_1\dash\uw$ because $p$ remains a tied winner. And so, $I$ is not in $\cale\dash\cc\dash\cald_2\dash\uw$ and not in $\cale\dash\dc\dash\cald_2\dash\nuw$.\qed
\end{proof}

Figure~\ref{fig:group4} shows the pairwise relationships that are not determined by Theorem~\ref{t:all-cc-dc}, and the rest of this section is focused on proving the relationships highlighted within that figure.

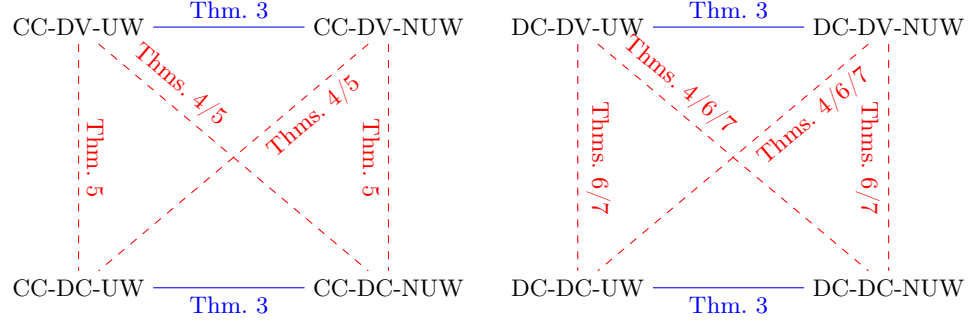
\begin{figure}[ht]
  \centering
  \raggedright
  \footnotesize
\caption{Relationship diagrams for Group~4, along with references to the theorems justifying those relationships. Solid lines indicate relationship we completely characterized, and dashed lines indicate we give sufficient conditions for incomparability.}\label{fig:group4}
  \begin{minipage}{0.4\linewidth}
  \begin{tikzpicture}%
  \node (dvuw)  {CC-DV-UW};
  \node[right= 2cm of dvuw] (dvnuw) {CC-DV-NUW};
  \node[below =3cm of dvuw] (dcuw) {CC-DC-UW};
  \node[below =3cm of dvnuw] (dcnuw) {CC-DC-NUW};
  \draw [color=red,dashed] (dvuw) -- (dcuw) node[midway, left,sloped, above]{Thm.~\ref{t:sufficient-group4a}};
  \draw [color=red,dashed] (dvnuw) -- (dcnuw) node[midway, right,sloped, below]{Thm.~\ref{t:sufficient-group4a}};
  \draw [color=blue] (dvuw) -- (dvnuw) node[midway,above]{Thm.~\ref{t:characterize-uw-nuw}};
  \draw [color=blue] (dcuw) -- (dcnuw) node[midway,below]{Thm.~\ref{t:characterize-uw-nuw}};
  \draw [color=red,dashed] (dvuw) -- (dcnuw) node[near start,sloped,above]{Thms.~\ref{t:sufficient-group4}/\ref{t:sufficient-group4a}};
  \draw [color=red,dashed] (dvnuw) -- (dcuw) node[near start, sloped,below]{Thms.~\ref{t:sufficient-group4}/\ref{t:sufficient-group4a}};
  \end{tikzpicture}
  \end{minipage}
  \hspace*{1.5cm}
  \begin{minipage}{0.4\linewidth}
  \begin{tikzpicture}%
  \node (dvuw)  {DC-DV-UW};
  \node[right= 2cm of dvuw] (dvnuw) {DC-DV-NUW};
  \node[below =3cm of dvuw] (dcuw) {DC-DC-UW};
  \node[below =3cm of dvnuw] (dcnuw) {DC-DC-NUW};
  \draw [color=red,dashed] (dvuw) -- (dcuw) node[midway, left,sloped, above]{Thms.~\ref{t:sufficient-group4bi}/\ref{t:sufficient-group4bii}};
  \draw [color=red,dashed] (dvnuw) -- (dcnuw) node[midway, right,sloped, below]{Thms.~\ref{t:sufficient-group4bi}/\ref{t:sufficient-group4bii}};
  \draw [color=blue] (dvuw) -- (dvnuw) node[midway,above]{Thm.~\ref{t:characterize-uw-nuw}};
  \draw [color=blue] (dcuw) -- (dcnuw) node[midway,below]{Thm.~\ref{t:characterize-uw-nuw}};
  \draw [color=red,dashed] (dvuw) -- (dcnuw) node[near start,sloped,above]{Thms.~\ref{t:sufficient-group4}/\ref{t:sufficient-group4bi}/\ref{t:sufficient-group4bii}};
  \draw [color=red,dashed] (dvnuw) -- (dcuw) node[near start, sloped,below]{Thms.~\ref{t:sufficient-group4}/\ref{t:sufficient-group4bi}/\ref{t:sufficient-group4bii}};
  \end{tikzpicture}
  \end{minipage}
\end{figure}

A voting rule is said to satisfy the majority criterion if it selects exactly the majority winner when one exists.
Below we give a sufficient condition for incomparability of the constructive types in Group~4, and that condition will suffice to prove our results about concrete voting rules.

\begin{theorem}\label{t:sufficient-group4a}
    Let $\cale$ be a voting rule that satisfies the majority criterion. Then each of $\cale\dash\cc\dash\dc\dash\uw$ and $\cale\dash\cc\dash\dc\dash\nuw$ are incomparable with each of $\cale\dash\cc\dash\allowbreak\dv\dash\uw$ and $\cale\dash\cc\dash\dv\dash\nuw$.
\end{theorem}
\begin{proof}
    For brevity, let $\calt_1 = \cale\dash\cc\dash\dc\dash\uw$, let $\calt_2 = \cale\dash\cc\dash\dc\dash\nuw$, let $\cals_1 = \cale\dash\cc\dash\allowbreak\dv\dash\uw$, and let $\cals_2 = \cale\dash\cc\dash\dv\dash\nuw$. It follows from definitions that $\calt_1 \subseteq \calt_2$ and $\cals_1 \subseteq \cals_2$. Therefore, it suffices to show that $\calt_1 \not\subseteq \cals_2$ and $\cals_1 \not\subseteq \calt_2$.
    
    We define the $\calt_1$ instance $I = (C, V, p, k)$ as follows: $C = \{a,b\}$, $V= \{a>b, a>b\}$, $p=b$, and $k=1$.
    Since $a$ is the majority winner of $(C, V)$, $b$ is not a unique winner. However, $b$ can be made a unique winner by deleting $a$. In other words, $I \in \calt_1$. On the other hand, regardless of which vote is deleted, $a$ remains a (unique) majority winner, so $I \not\in \cals_2$.
    Therefore, $\calt_1 \not\subseteq \cals_2$.

    We now define the $\cals_1$ instance $I' = (C', V', p', k')$ as follows: $C' = \{a,b,c,d\}$, $V'=\{a>b>c>d, a>b>c>d, d>a>b>c\}$, $p'=d$, and $k'=2$.
    Deleting the two voters that rank $a$ first makes $d$ a unique (majority) winner, so $I' \in \cals_1$. On the other hand, deleting any two candidates from $\{a,b,c\}$ leaves the remaining candidate in that set as a (unique) majority winner, so $I' \not\in \calt_2$.
    \qed
\end{proof}

We are now left to consider the destructive types in Group~4, and we show a similar result by appealing to the majority criterion.

\begin{theorem}\label{t:sufficient-group4bi}
    Let $\cale$ be a voting rule that satisfies the majority criterion. 
    Each of $\cale\dash\dc\dash\allowbreak\dv\dash\uw$ and $\cale\dash\dc\dash\dv\dash\nuw$ are not contained in each of 
    $\cale\dash\dc\dash\dc\dash\uw$ and $\cale\dash\dc\dash\dc\dash\nuw$.
\end{theorem}
\begin{proof}
    For brevity, let $\calt_1 = \cale\dash\dc\dash\dc\dash\uw$, let $\calt_2 = \cale\dash\dc\dash\dc\dash\nuw$, let $\cals_1 = \cale\dash\dc\dash\allowbreak\dv\dash\uw$, and let $\cals_2 = \cale\dash\dc\dash\dv\dash\nuw$. It follows from definitions that $\calt_2 \subseteq \calt_1$ and $\cals_2 \subseteq \cals_1$. Therefore, it suffices to show that
    $\cals_2 \not\subseteq \calt_1$.

    Define the $\cals_2$ instance $I' = (C', V', p', k')$ as follows: $C' = \{a,b,c,d\}$, $V'=\{a>b>c>d, a>b>c>d, d>a>b>c\}$, $p'=a$, and $k'=2$.
    Deleting the two voters that rank $a$ first makes $d$ a unique (majority) winner, so $I' \in \cals_2$. On the other hand, deleting any two candidates from $\{b,c, d\}$ leaves $a$ as a (unique) majority winner, so $I' \not\in \calt_1$.
    \qed
\end{proof}

We note that the majority criterion is \emph{not} necessary for the $\not\subseteq$ relationships established by Theorems~\ref{t:sufficient-group4a} and~\ref{t:sufficient-group4bi}; Carleton et al.~\cite{car-cha-hem-nar-tal-wel:j:sct} show the same incomparabilities hold under Veto, but Veto (clearly) does not satisfy the majority criterion. 

For each $\not\subseteq$ relationship given by Theorem~\ref{t:sufficient-group4bi}, attempting to prove the corresponding $\not\supseteq$ relationship is not possible as Carleton et al.~\cite{car-cha-hem-nar-tal-wel:j:sct} show that under approval---which satisfies the majority criterion~\cite{las-san:b:approval}---DC-DC-UW is contained in DC-DV-UW and DC-DC-NUW is also contained in DC-DV-NUW\@. More specifically, they prove that the first containment holds under every voting rule that satisfies Property Unique~$\alpha$, while the second containment holds under every voting rule that satisfies Property~$\alpha$.\footnote{
A voting rule $\cale$ satisfies Property~$\alpha$~\cite{sen:j:rational-choice} (resp., Unique-$\alpha$~\cite{hem-hem-rot:j:destructive-control}) if $p$ being a winner (resp., unique winner) of election $(C ,V)$ implies that $p$ remains a winner (resp., unique winner) in each election $(C-C', V)$, where $C' \subseteq C-\{p\}$.} We would ideally like to prove the converse of their result (thus obtaining a complete characterization), but we instead make an additional assumption to prove the corresponding separations. 

We say that an election $(C, V)$ is 1-abstention safe under $\cale$ if for each vote $v \in V$, $\cale(C, V-\{v\}) = \cale(C, V)$. This essentially captures the fact that even if a voter in $V$ were to abstain, the outcome of the election would not change.

We say a voting rule $\cale$ fails to satisfy Property~$\alpha$ (resp., Unique-$\alpha$) via a 1-abstention-safe election if there is a 1-abstention safe election $(C, V)$ such that $p$ is a winner (resp., unique winner) of $(C, V)$, but there is a set $C' \subseteq C-\{p\}$ such that $p$ is not a winner (resp., unique winner) of $(C-C', V)$.

The above condition may seem unnatural and unusually restrictive, but as it turns out, it is one satisfied by every concrete voting rule used in this paper (and we make that clear each time we cover a concrete voting rule).
Moreover, this condition has the advantage that looking for separation witnesses for pairwise relationships it characterizes can be simplified to finding witnesses that are 1-abstention safe under a given voting rule.

\begin{theorem}\label{t:sufficient-group4bii}
    Let $\cale$ be a voting rule.
    \begin{enumerate}
    \item If $\cale$ fails to satisfy Property~$\alpha$ via a 1-abstention-safe election, then $\cale\dash\dc\dash\dc\dash\nuw \not\subseteq \cale\dash\dc\dash\dv\dash\nuw$.
    \item If $\cale$ fails to satisfy Property Unique-$\alpha$ via a 1-abstention-safe election, then $\cale\dash\dc\dash\dc\dash\uw \not\subseteq \cale\dash\dc\dash\dv\dash\uw$.
    \end{enumerate}
\end{theorem}
\begin{proof}

    Suppose $\cale$ fails to satisfy Property~$\alpha$ via a 1-abstention-safe election. Then there is a 1-abstention safe election $(C, V)$ under $\cale$ such that $p$ is a winner of $(C, V)$, but there is a set $C' \subseteq C-\{p\}$ such that $p$ is not a winner of $(C-C', V)$. Moreover, we can assume without loss of generality that $\card{C'} +1 = \card{C}$. Indeed, if that were not the case, then $p$ would remain a winner in every election $(\hat{C}, V)$ where $\card{\hat{C}} = \card{C'}+1$, in which case using any $\hat{C}$ as $C$ would satisfy the assumption.
    It follows from the definitions for the properties assumed that $(C, V, p, 1) \in \cale\dash\dc\dash\dc\dash\nuw$ but not in $\cale\dash\dc\dash\dv\dash\nuw$.

    Suppose $\cale$ fails to satisfy Property Unique-$\alpha$ via a 1-abstention-safe election. Then there is a 1-abstention safe election $(C, V)$ under $\cale$ such that $p$ is a unique winner of $(C, V)$, but there is a set $C' \subseteq C-\{p\}$ such that $p$ is not a unique winner of $(C-C', V)$. Moreover, we can again assume without loss of generality that $\card{C'} +1 = \card{C}$ by using a similar argument as in the proof of the previous part.
    It follows from the definitions for the properties assumed that $(C, V, p, 1) \in \cale\dash\dc\dash\dc\dash\uw$ but not in $\cale\dash\dc\dash\dv\dash\uw$.
    \qed
\end{proof}

\newcommand{\spacecol}{@{\ \ }c@{\ \ }}

\begin{table}[h]
    \centering
    \caption{Summary of our results about concrete voting rules. The total number of pairs for each line sums up to 322. 
    Our new results are boldfaced. More details for each voting rule can be found in the appendix.\\
    ${}^\dagger$: The characterizations needed to see the 9 nonuniversal collapsing pairs here were established by \cite{fal-hem-hem-rot:j:llull}, so we attribute these collapses to them even though those are not explicitly noted in their paper.}
    \label{tab:concrete-summary}
    \begin{tabular}{|\spacecol||\spacecol|\spacecol||\spacecol|\spacecol|\spacecol|}
        \multirow{2}{*}{Voting Rule} & \multirow{2}{*}{Collapses} & \multirow{2}{*}{Separations} & \multicolumn{3}{\spacecol|}{Subclassifications of Separations}\\
         {}& {} &{} & $\subsetneq$ or $\supsetneq$ & Incomparable & Open\\\hline
         $k$-NRV, $k\geq 2$ & 7 & 213 + \textbf{102} & \textbf{38} & 213 + \textbf{64} & 0\\ 
         Copeland& 7 + 9${}^\dagger$ & \textbf{302} & \textbf{32} & \textbf{270} & 4\\ 
         Llull & 7 + 9${}^\dagger$  & \textbf{303} & \textbf{36} & \textbf{267} & 3\\ 
         Schulze & 7 + \textbf{9} & \textbf{304} & \textbf{30} & \textbf{274} & 2\\ 
         Ranked Pairs & 7 + \textbf{55} & \textbf{260} & \textbf{0} & \textbf{260} & 0\\
         Bucklin &  7 & \textbf{315} & \textbf{38} & \textbf{277} & 0\\ 
         Fallback &  7 & \textbf{315} & \textbf{38} & \textbf{277} & 0
    \end{tabular}
\end{table}

\section{Conclusions and Future Work}\label{sec:conclusion}

We have thus given for the first time universal separation results, and we have completely determined under which conditions compatible pairs of control types within Groups 1--3 collapse. 
We gave various sufficient conditions for the incomparability of pairs of control types in Group~4, and explored the relationships between pairs of control types within Group~5 by studying concrete voting rules; those results are summarized in Table~\ref{tab:concrete-summary}.
Using a general result (i.e., Theorem~\ref{t:universal-separations}), we also determined that no two compatible control types can collapse if one is about constructive control and the other is about destructive control; this result holds within Groups 1--5.

As Table~\ref{tab:concrete-summary} shows, this paper has obtained 64 new collapse results and 1901 new separation results. Of those 1901 new separations, more than half follow from the theorems in Section~\ref{sec:nonpartition}, and for all but nine of the 1901 we have obtained a precise subclassification.

As future work, we suggest strengthening our axiomatic results to broader ones or even perhaps to complete characterizations. We also commend to the reader the challenge of the nine open subclassification cases for Copeland, Llull, and Schulze. There is much to be done to understand necessary and sufficient conditions for collapses between control types in Group~5.

\begin{credits}
\subsubsection{Acknowledgments} \acktext
\end{credits}

\bibliographystyle{splncs04}

\clearpage

\newcommand{\mainUnit}{\section}
\newcommand{\secondaryUnit}{\subsection}

\appendix

\mainUnit{Deferred Definitions for Partitioning}\label{app:defs}

A partition of a candidate set $C$ is a pair $(C_1, C_2)$ such that $C_1 \cap C_2 = \emptyset$ and $C_1 \cup C_2 = C$. A partition of a collection of voters $V$ is a pair $(V_1, V_2)$ of subcollections of $V$ such that $V_1 \cup V_2 = V$, where $\cup$ here denotes multiset union.
Moreover, we at times speak of an election $(C', V)$ whose votes are over a set $C \supseteq C'$. In such a case, as is standard in the literature, we treat the votes in $V$ to be ``masked down'' to be over $C'$.

In the partition-based control types below, the control action is carried out in two rounds/stages: the first stage and the final stage/round. The first stage consists of at most two subelections (whose candidates/votes depend on the partitioning action). Because there may be multiple winners in a subelection, a tie-breaking rule is specified as part of the definition of each partition-based control problem. The two standard tie-breaking rules, which are the ones we  use in this paper, are the ties-eliminate (TE) handling rule that eliminates tied winners in a subelection (and thus only a unique winner of a subelection survives the TE tie-handling rule) and the ties-promote (TP) handling rule that eliminates none of the winners (and thus every winner of a subelection survives the TP tie-handling rule). Subelection election winners that survive the tie-handling rule then proceed to the final round, for one last election.
The original work on control by Bartholdi, Tovey, and Trick~\cite{bar-tov-tri:j:control} ignored the issue of ties, but \cite{hem-hem-rot:j:destructive-control} established the TE and TP models and showcased how control complexity can change based on the tie-breaking rule used.

\begin{definition}[\cite{hem-hem-men:j:search-versus-decision,car-cha-hem-nar-tal-wel:j:sct}]\label{def:control-partition}  
Let $\cale$ be a voting rule.
\begin{enumerate}    
    \item In the \textbf{constructive control by partition of voters} problem for $\cale$, in the TP or TE tie-handling rule model
    (denoted by $\cale\dash\cc\dash\pv\dash\tp\dash\nuw$ or $\cale\dash\cc\dash\pv\dash\te\dash\nuw$, respectively), 
    we are given an election $(C, V)$, and a candidate $p \in C$. We ask if there is a 
    partition
    of $V$ into $V_1$ and $V_2$ such that $p$ is a winner of the two-stage election where the winners of
   subelection $(C, V_1)$ that survive the tie-handling rule compete 
   (with respect to vote collection $V$)
    along with
    the winners of subelection $(C, V_2)$ that survive the tie-handling rule.
    Each election (in both stages) is conducted using election system $\cale$.
    
    \item In the \textbf{constructive control by run-off partition of candidates} problem for $\cale$, in the TP or TE tie-handling rule model
    (denoted by $\cale\dash\cc\dash\rpc\dash\tp\dash\nuw$ or $\cale\dash\cc\dash\rpc\dash\te\dash\nuw$, respectively), 
    we are given an election $(C, V)$, and a candidate $p \in C$. We ask if there is a 
    partition of $C$ into $C_1$ and $C_2$ such that $p$ is a winner of the two-stage
    election where the winners of
    subelection $(C_1, V)$ that survive the tie-handling
    rule compete 
    (with respect to vote collection $V$)
    against the winners of 
    subelection $(C_2, V)$
    that survive the tie-handling rule.
    Each election (in both stages) is conducted using election system $\cale$.
    
    \item In the \textbf{constructive control by partition of candidates} problem for $\cale$, in the TP or TE tie-handling rule model
    (denoted by $\cale\dash\cc\dash\pc\dash\tp\dash\nuw$ or $\cale\dash\cc\dash\pc\dash\te\dash\nuw$, respectively), 
    we are given an election $(C, V)$, and a candidate $p \in C$. We ask if there is a 
    partition of $C$ into $C_1$ and $C_2$ such that $p$ is a winner of the two-stage
    election where the winners of
    subelection $(C_1, V)$ that survive the tie-handling
    rule compete (with respect to vote collection $V$)
    against 
    all candidates in $C_2$.
    Each election (in both stages) is conducted using election system $\cale$.
    \end{enumerate}
\end{definition}

{\section{Results for Concrete Voting Rules}\label{sec:results}

In this section, we give results about Normalized Range Voting, Copeland, Llull, Schulze, Ranked Pairs, Bucklin, and Fallback\@. For each rule we determine which of the 322 pairs separate and which collapse, and we refine our separation results by proving either strict containments or incomparability. 

We prove new collapses and containments under Llull, Copeland, Schulze, and Ranked Pairs. 
We also prove that under $k$-NRV, for each $k \geq 2$ (the $k=1$ case is simply approval, which was studied by Carleton et al.~\cite{car-cha-hem-nar-tal-wel:j:sct}), under Bucklin, and under Fallback, there are no additional collapses or containments beyond the universal ones.
We also prove under each voting rule above that no other containments (and so certainly also collapses) beyond the ones we find exist. A summary of our results is presented in Table~\ref{tab:concrete-summary}.

Before turning to results about individual voting rules, we remind the reader of the previously-known universal collapses and containments, which will be used in this section. Hemaspaandra, Hemaspaandra, and Menton~\cite{hem-hem-men:j:search-versus-decision}
established seven universal collapsing pairs by showing that for each voting rule $\cale$, it holds that $\cale\dash\mathrm{DC\text{-}RPC\text{-}TE\text{-}NUW} = \cale\dash\mathrm{DC\text{-}PC\text{-}TE\text{-}NUW} = \cale\dash\mathrm{DC\text{-}RPC\text{-}TE\text{-}UW} = \cale\dash\mathrm{DC\text{-}PC\text{-}TE\text{-}UW}, \text{ and also that }
\cale\dash\mathrm{DC\text{-}RPC\text{-}TP\text{-}NUW} = \cale\dash\mathrm{DC\text{-}PC\text{-}TP\text{-}NUW}.$
Moreover, Carleton et al.~\cite{car-cha-hem-nar-tal-wel:j:sct} further proved that these are the only universal collapses, and they also proved the following universal containments for each voting rule $\cale$,
$\cale\dash\mathrm{DC\text{-}RPC\text{-}TP\text{-}UW}\subseteq \cale\dash\mathrm{DC\text{-}RPC\text{-}TE\text{-}NUW} \text{ and }\cale\dash\mathrm{DC\text{-}PC\text{-}TP\text{-}UW}\subseteq \cale\dash\mathrm{DC\text{-}RPC\text{-}TE\text{-}NUW}.$

It will be clear that each voting rule in this paper satisfies the condition of Theorem~\ref{t:all-cc-dc}, and that Ranked Pairs is the only voting rule that is weakly resolute in this pair. This establishes all relationships between all pairs from Groups 1--3.

For Group~4, we will show that each voting rule in this section satisfies the majority criterion and fails to satisfy Properties~$\alpha$ and Unique-$\alpha$ via a 1-abstention-safe election, thus determining the relationships between all pairs from Group~4. To do so, we will appeal to the generalization of McGarvey's Theorem given by Hemaspaandra, Lavaee, and Menton~\cite{hem-lav-men:j:schulze-and-ranked-pairs}, which is a result relating elections and so-called weighted majority graphs, which we define below.

Let $E=(C, V)$ be an election. For two distinct candidates $a$ and $b$, $N_E(a,b)$ denotes the number of votes in $V$ that rank $a$ above $b$. The majority margin $m_E(a,b) = N_E(a,b) - N_E(b,a)$
represents the difference between the number of voters who prefer $a$ over $b$ and those who prefer $b$ over $a$.
We drop the $E$ subscript in the above notation when the election is clear from context.
A weighted majority graph (WMG) for $(C, V)$ is a weighted directed graph $G$, where the vertices are the candidates in $C$ and for two distinct vertices $a$ and $b$, there is an edge $(a,b)$ with weight $m(a,b)$.
McGarvey's Theorem, in particular the generalization of it given by~\cite{hem-lav-men:j:schulze-and-ranked-pairs}, gives a polynomial-time algorithm to construct from a weighted majority graph an election whose pairwise contests (and margins) are exactly those specified in the graph. We will appeal to this result multiple times in this section.

Finally, to handle the Group~5 results, we in each subsection, either prove a containment/separation about the corresponding voting rule or point to a table containing the separation witnesses.

\subsection{$k$-Normalized Range Voting}\label{subsec:kNRV}

For each integer $k \geq 1$,
$k$-Normalized Range Voting ($k$-NRV) is a normalized version of the so-called $k$-range voting~\cite{men:j:range-voting}. 
In an election $(C, V)$ of $k$-NRV, each vote is a $\card{C}$-tuple of integers from $\{0, \ldots, k\}$, representing the fact that
each voter assigns each candidate a raw score from $\{0,1,\dots,k\}$. Before accumulating the scores, we first check the highest score $M$ and lowest score $m$ on each ballot. If $M=m$, the ballot does not show any relative preference between the candidates and does not distinguish between them and can therefore ignored. Otherwise, each score $s$ on the ballot will be converted to $k(s-m)/(M-m)$. Note that after normalization, the highest score is mapped to $k$, the lowest score is mapped to $0$, and all other scores are linearly adjusted. The candidate with the highest accumulated total score wins.

For example, when $k=5$ and $C=\{a,b,c\}$, the ballot $(4,1,0)$ is normalized to $(5,5/4,0)=(5,1.25,0)$, while the ballot $(2,2,2)$ is ignored because it assigns the same score to each candidate.

In this section, we assume $k \geq 2$, because, as Menton~\cite{men:j:range-voting} mentions, when $k=1$ the normalization step of $k$-NRV actually has no effect, so $1$-NRV is equivalent to approval voting; based on that, the existing results of Carleton et al.~\cite{car-cha-hem-nar-tal-wel:j:sct} regarding approval already establish the results needed for 1-NRV.

We prove that for each $k \geq 2$, the only possible containments (and thus collapses) under $k$-NRV are the ones that hold universally. The list of separations can be found in Table~\ref{tab:2nrv-relationships}. Our approach leverages the fact that for each positive integers $k_1$ and $k_2$ where $k_1 \leq k_2$, we have a nifty connection between $k_1$-NRV and $k_2$-NRV, and we use that connection to study the infinite collection of voting rules in the NRV family of voting rules; in particular, the winners of a $k_1$-NRV election remain winners of the same election under $k_2$-NRV, so we can reuse separation witnesses from the study of approval by Carleton et al.~\cite{car-cha-hem-nar-tal-wel:j:sct} and supplement those with the new separation witnesses that we compute by only studying 2-NRV to prove that no other containments/collapses hold under $k$-NRV, for each $k \geq 2$, beyond the universal ones.

\begin{theorem}\label{thm:knrv-invariance}
	Let $k_1$ and $k_2$ be positive integers, and assume $k_1\le k_2$. 
    Let $(C,V)$ be a $k_1$-NRV election. 
    Then the winner set of election $(C, V)$ under
	$k_1$-NRV coincides with that under $k_2$-NRV.\footnote{%
		In fact, whether a profile can be reused depends only on the raw scores it actually uses. For example, a $5$-NRV instance using only scores from $\{0,1,2\}$ can equally be viewed as a $2$-NRV instance, and hence can be reused for every $k$-NRV with $k\ge 2$.}
\end{theorem}
\begin{proof}
	Fix a voter $v\in V$ and a candidate $c\in C$. Let $s_v(c)$ be the raw score of $c$ from ballot $v$. For $k\in\{k_1,k_2\}$ we will define the normalized score of $c$ from ballot $v$ under $k$-NRV (denoted $s_v^k(c)$) below. Let $M_v$ and $m_v$ be the highest and lowest raw scores on ballot $v$.
	
	If $M_v=m_v$, then $v$ is ignored. Otherwise,
	\[
	\text{let }s_v^{k_1}(c)=k_1\cdot\dfrac{s_v(c)-m_v}{M_v-m_v}
	\quad\text{and}\quad
	\text{let }s_v^{k_2}(c)=k_2\cdot\dfrac{s_v(c)-m_v}{M_v-m_v}.
	\]

    Therefore, $s_v^{k_2}(c)=\dfrac{k_2}{k_1}\,s_v^{k_1}(c).$	
    For $k \in \{k_1, k_2\}$, we will define $S^{k}(c)$ to be the score of candidate $c$ in election $(C, V)$ under $k$-NRV\@. And so,
	summing over all $v\in V$ (ignoring ballots contributing $0$ points), we get for each candidate $c$ that
	\[
	S^{k_2}(c)
	= \sum_{v\in V} s_v^{k_2}(c)
	= \frac{k_2}{k_1}\sum_{v\in V} s_v^{k_1}(c)
	= \frac{k_2}{k_1}\,S^{k_1}(c).
	\]
	
	Since $k_2/k_1>0$, multiplying each candidate's total score by the same positive
	number does not change any strict or weak comparisons between them. Therefore, the score ranking, including
	ties, is the same under $k_1$-NRV and $k_2$-NRV. Thus, the winner sets are the same.\qed
\end{proof}

\begin{corollary}\label{cor:knrv-separation-transfer}
    Let $k$ and $k'$ be a positive integers such that $k\leq k'$, and let $\calt_1$ and $\calt_2$ be two compatible control actions.
    If there is an $I \in k\dash\nrv\dash\calt_1 - k\dash\nrv\dash\calt_2$ (and thus $k\dash\nrv\dash\calt_1 \not\subseteq k\dash\nrv\dash\calt_2$), then $I \in k'\dash\nrv\dash\calt_1 - k'\dash\nrv\dash\calt_2$.
\end{corollary}
\begin{proof} 
    Assume there is an $I \in k\dash\nrv\dash\calt_1 - k\dash\nrv\dash\calt_2$.
	Any action on the election in instance $I$ (adding, deleting, or partitioning candidates or voters) will only result elections using the candidates and voters already existing in $I$.
    Therefore, the raw scores appearing in these elections still belong to $\{0,\dots,k\}\subseteq\{0,\dots,k'\}$.
		
	By Theorem~\ref{thm:knrv-invariance}, each such election has the same winner set under $k$-NRV and $k'$-NRV. That is, changing the parameter from $k$ to $k'$ will not change the original results of these control types. Therefore, 
    $I$ is also in $k'\dash\nrv\dash\calt_1$ but not in $k'\dash\nrv\dash\calt_2$.\qed
\end{proof}

In our search for separation witnesses, every separation witness determined for approval by~\cite{car-cha-hem-nar-tal-wel:j:sct} can be used to separate to corresponding pair under $k$-NRV, for each $k \geq 2$. Moreover, we prove that beyond the collapses and containments that hold in the universal case, there are no other containments that hold under 2-NRV; and so, the corresponding statement holds, for each $k' \geq 3$, under $k'$-NRV.

Finally, we establish the properties of 2-NRV we need to complete the picture with respect to Group~4. But before doing so, we mention that because votes in 2-NRV are not linear orders, we must clarify what it means for a candidate to be ``most preferred'' by a majority of voters. In this sense, it means that such a candidate will (uniquely) have the highest score in that vote, and thus after normalization, that candidate will be the only one gaining 2 points from that vote.

\begin{proposition}\label{prop:kNRV-group4}
    2-NRV satisfies the majority criterion and fails to satisfy Properties~$\alpha$ and Unique-$\alpha$ via a 1-abstention-safe election. 
\end{proposition}
\begin{proof}
    Let $(C, V)$ be a 2-NRV election and assume a candidate $p$ is most preferred by at least $k = \lfloor\card{V}/2\rfloor +1$ voters. Then the score of that candidate is at least $2k$. 
    Moreover, observe that $\card{V}-k < k$.
    And so, the highest score another candidate can receive is at most $(\card{V}-k) + k < 2k$. Therefore $p$ is a unique winner under 2-NRV.

    Let $C = \{a,b,c\}$ and let $V$ contain the following 14 votes (the values on the left side of the $\implies$ are the raw scores, and the values on the right side of the $\implies$ are the normalized scores):
    \[
    \begin{array}{cccc}
         v_1:& 1, 1, 0 &\quad\implies\quad& 2, 2, 0\\
         v_2:& 0, 2, 1 &\quad\implies\quad&0, 2, 1\\
         v_3:& 1, 2, 0&\quad\implies\quad&1, 2, 0\\
         v_4:& 2, 1, 1&\quad\implies\quad&2, 0, 0\\
         v_5:& 2, 0, 1&\quad\implies\quad&2, 0, 1\\
         v_6:& 0, 1, 1&\quad\implies\quad&0, 2, 2\\
         v_7:& 1, 0, 1&\quad\implies\quad&2, 0, 2\\
         v_8:& 0, 1, 2&\quad\implies\quad&0, 1, 2\\
         v_9:& 1, 2, 0&\quad\implies\quad&1, 2, 0\\
         v_{10}:& 1, 2, 0&\quad\implies\quad&1, 2, 0\\
         v_{11}:& 0, 1, 2&\quad\implies\quad&0, 1, 2\\
         v_{12}:& 2, 1, 1&\quad\implies\quad&2, 0, 0\\
         v_{13}:& 1, 0, 0&\quad\implies\quad&2, 0, 0\\
         v_{14}:& 2, 0, 1&\quad\implies\quad&2, 0, 1\\
    \end{array}
    \]

    In $(C, V)$, $a$ receives 17 points, $b$ receives 14 points, and $c$ receives 11 points. Thus $a$ is a unique winner. In the election $(\{a,b\}, V)$, it can be verified that $a$ receives 8 points and $b$ receives 14 points, thus making $b$ a unique winner (the corresponding normalized score vectors have been omitted). 
    Thus 2-NRV fails to satisfy Properties~$\alpha$ and Unique-$\alpha$.
    Moreover deleting a single vote from election $(C, V)$ can decrease $a$'s score by at most 2 points and cannot increase the scores of $b$ and $c$, thus $(C, V)$ is 1-abstention safe under 2-NRV\@.
    \qed
\end{proof}

\subsection{Copeland}\label{subsec:copeland-05}

Copeland$^\alpha$, $\alpha \in [0,1]$, is a family of voting rules based on pairwise comparisons of candidates, where each candidate receives one point for each pairwise contest (aka, head-to-head contest) they uniquely win against the other candidates, zero points for each such contest in which they lose, and $\alpha$ points for each such contest in which they tie. A Copeland$^{\alpha}$ winner is a candidate with maximal score.
Let $E=(C, V)$ be an election.
For distinct $a, b \in C$, 
if $N_E(a, b) > N_E(b, a)$, then $a$ is said to pairwise defeat $b$ (denoted as $a\succ_E b$); if $N_E(a, b) = N_E(b, a)$, then $a$ and $b$ are said to be pairwise tied (denoted as $a\sim_E b$). For each $\alpha \in [0,1]$, the Copeland$^\alpha$ score of candidate $a\in C$ in election $(C, V)$ is defined as
	\[
	\operatorname{score}_E^{\alpha}(a)
	=
	\card{\{\,b\in C\setminus\{a\}\mid a\succ_E b\,\}}
	\;+\;
	\alpha\cdot
	\card{\{\,b\in C\setminus\{a\}\mid a\sim_E b\,\}}.
	\]

When using the abovedefined notation, we at times drop the subscript $E$ when the election is clear from context.

In this subsection, we study the Copeland voting rule, which is simply Copeland$^{0.5}$. We give an example of winner evaluation under Copeland below.

\begin{example}\label{ex:copeland-vs-llull}
	Let $C=\{a,b,c\}$, and consider the following vote collection $V$:
	\[
	\begin{array}{cccc}
		v_1: & a > b > c, \\
		v_2: & a > c > b, \\
		v_3: & b > a > c, \\
		v_4: & c > b > a.	
	\end{array}	
	\]

    Let $E=(C, V)$. 
	The pairwise results are:	
	$a \sim_{E} b$, %
	$b \sim_E c$, %
	and $a \succ_E c$. %
	For $\alpha=0.5$, $\operatorname{score}_E^{0.5}(a)=1+\tfrac{1}{2}=1.5$,	
	$\operatorname{score}_E^{0.5}(b)=\tfrac{1}{2}+\tfrac{1}{2}=1.0$,	
	and $\operatorname{score}_E^{0.5}(c)=\tfrac{1}{2}$, so $a$ is the only Copeland winner.
\end{example}
If a candidate $p$ wins in every pairwise matchup against every other candidate, then $p$ is a Condorcet winner. If $p$ does not lose in every pairwise matchup, then $p$ is a weak Condorcet winner. Copeland$^\alpha$ is known to be Condorcet consistent, i.e., it elects a Condorcet winner (uniquely) if one exists~\cite{fal-hem-hem-rot:j:llull}.

We first address the Group~4 relationships.

\begin{proposition}\label{prop:copeland-group4}
    For each $\alpha \in [0,1]$,
    Copeland$^{\alpha}$
    satisfies the majority criterion and fails to satisfy Properties~$\alpha$ and Unique-$\alpha$ via a 1-abstention-safe election. 
\end{proposition}
\begin{proof}
    Let $\alpha \in [0,1]$.
    The satisfaction of the majority criterion is immediate as a majority winner is also a Condorcet winner, and thus is always a unique winner under Copeland$^\alpha$.

    We will now construct a Copeland$^{\alpha}$ election $(C, V)$ by first specifying requirements that the election must satisfy. Let $C=\{x_1, x_2, x_3, p, d\}$. For each pairwise contest specified below, the margin is 3.
    \begin{itemize}
        \item For each $i \in [3]$, $p$ pairwise defeats $x_i$;
        \item $d$ pairwise defeats $p$;
        \item for each $i \in [3]$, $x_i$ pairwise defeats $d$; and
        \item $x_1$ pairwise defeats $x_2$, $x_2$ pairwise defeats $x_3$, and  $x_3$ pairwise defeats $x_1$.
    \end{itemize}

    Let $(C, V)$ be the election constructed using the weighted version of McGarvey's Theorem.
    First notice that $p$ is a unique winner of $(C, V)$ with 4 points. (Notice that $p$ receives one point for each $c \neq p$ where $m(p, c) > 0$.) By deleting $x_1$, $x_2$, $x_3$, and $x_4$, $d$ becomes a unique winner. This demonstrates that Copeland$^\alpha$ fails to satisfy Properties~$\alpha$ and Unique-$\alpha$. Moreover, deleting a single vote can decrease each positive margin by at most one, but doing so does not change the pairwise contests in the majority graph for $(C, V)$. Thus $(C, V)$ is 1-abstention safe under Copeland$^\alpha$.
    \qed
\end{proof}

For each $\alpha\in[0,1)$, Faliszewski et al.~\cite{fal-hem-hem-rot:j:llull} determined the complexity of destructive control by partitioning candidates under Copeland$^\alpha$. Their proof shows that all six destructive candidate partitions are governed by the same condition: in each Yes instance $(C, V, p)$, the distinguished  candidate $p$ is not a Condorcet winner of election $(C, V)$. Although they do not note it in their work, their observation implies that the six types they consider pairwise collapse. We formally state this below.

\begin{corollary}[\protect{\cite[Proof of Theorem 4.5]{fal-hem-hem-rot:j:llull}}]\label{cor:copeland-dc}
For each $\alpha\in [0,1]$,  under Copeland$^{\alpha}$,
$\mathrm{DC\text{-}PC\text{-}TE\text{-}UW} = \allowbreak
\mathrm{DC\text{-}PC\text{-}TE\text{-}NUW} = \allowbreak
\mathrm{DC\text{-}RPC\text{-}TE\text{-}UW} = \allowbreak
\dc\dash\allowbreak\rpc\dash\te\dash\nuw = \allowbreak	
\mathrm{DC\text{-}PC\text{-}TP\text{-}UW} = \allowbreak	
\mathrm{DC\text{-}RPC\text{-}TP\text{-}UW}$.
\end{corollary}
\begin{proof}
    In the Proof of Theorem~4.5 from \cite{fal-hem-hem-rot:j:llull}, they observe that for each of the problems listed the statement of the corollary, $(C, V, p)$ is a Yes-instance of that problem if and only if there is a $C'\subseteq C$ such that $p$ is not a unique winner of election $(C' \cup \{p\}, V)$  under Copeland$^{\alpha}$.\qed
\end{proof}

Naturally, the first four types listed above also pairwise collapse from the proof given by Hemaspaandra, Hemaspaandra, and Mentons's~\cite{hem-hem-men:j:search-versus-decision} universal collapses, but the corresponding result for Copeland$^\alpha$ was implicitly in the literature even before that!

Although the above collapses hold for all $\alpha \in [0,1)$, we only consider the case for $\alpha = 0.5$ in our search for separation witnesses and/or collapses.
Therefore, the separation witnesses given in Table~\ref{tab:copeland-relationships} only correspond to the case where $\alpha = 0.5$.\footnote{We mention in passing that some of our separation witnesses for pairs of control types that are about partitioning candidates have an odd number of voters. Therefore, ties do not occur in such any of the head-to-head contests, so such separation witnesses hold under Copeland$^\alpha$, for each $\alpha \in [0,1]$. Exploring separations/collapses under every voting rule in the Copeland$^\alpha$ family is beyond the scope of this paper, so we do not explore this fact further.}

\subsection{Llull}\label{subsec:llull}

Llull is member of the Copeland$^\alpha$ family corresponding to the case where $\alpha=1$.
For the election in Example~\ref{ex:copeland-vs-llull}, Llull gives
$\operatorname{score}_E^{1}(a)=2$,
$\operatorname{score}_E^{1}(b)=2$,
and $\operatorname{score}_E^{1}(c)=1$,
so the set of winners is $\{a,b\}$, while Copeland$^{0.5}$ has only one winner $a$.
That is to say, even if the profiles are exactly the same, the winners of Llull and Copeland$^{0.5}$ may be different.

In this section, we use the notation $a \succcurlyeq_E b$ to mean ``$a \succ_E b$ or $a \sim_E b$.''

Proposition~\ref{prop:copeland-group4} establishes the Group~4 relationships, so we focus on the Group~5 relationships now.

The collapses in Corollary~\ref{cor:copeland-dc} also hold under Llull as it is the $\alpha=1$ case. 
However, unlike Copeland, 
Llull exhibits containment relationships in the constructive case. Before giving our results, we prove a useful lemma above Llull (though the result holds for all voting rules in the Copeland$^\alpha$ family).

\begin{lemma}\label{lem:llull-remove-defeaters-preserve-winner}
    Let $\alpha \in [0,1]$ and let $(S, V)$ be a Copeland$^\alpha$ election with at least two candidates.
	Suppose $p$ is a winner of the Llull election $(S,V)$, and suppose $r\in S\setminus\{p\}$ defeats $p$ in pairwise contest.	
	Then $p$ is still a winner in $(S\setminus\{r\},V)$.	
	Furthermore, if $p$ is a unique winner of $(S,V)$, then $p$ is still a unique winner in $(S\setminus\{r\},V)$.	
	Therefore, for any set $R\subseteq S\setminus\{p\}$ consisting only of candidates who defeat $p$,	
	removing all candidates from $R$ still results in $p$ being a winner; if $p$ is originally a unique winner, then $p$ remains a unique winner.
\end{lemma}
\begin{proof}
    Fix $\alpha$, $(S, V)$, $p$ and $r$ to satisfy the requirements of the lemma.

    First consider the election $(S\setminus \{r\}, V)$ where $r$ is removed.
	Since $r$ defeats $p$, $p$ does not receive any points from $r$. Therefore, removing $r$ does not change $p$'s 
    Copeland$^{\alpha}$
    score.	
	For any other candidate $x \in S \setminus \{p,r\}$, removing $r$ only reduces the comparison between $x$ and $r$. Therefore, $x$'s score can only decrease or remain unchanged.	
	So, after removing $r$, $p$'s score remains unchanged, while the scores of the other candidates do not increase.	
	Thus, if $p$ is originally a winner, it remains a winner after removing $r$; if $p$ is originally a unique winner, it remains a unique winner after removing $r$.	
	For set $R$, remove each candidate and repeat the above argument.\qed
\end{proof}

\begin{theorem}\label{thm:llull-cc-rpc-te-uw-subset-pc-te-uw}
    Under Llull,
	$\mathrm{CC\text{-}RPC\text{-}TE\text{-}UW}
	\subseteq
	\mathrm{CC\text{-}PC\text{-}TE\text{-}UW}
	\subseteq
	\mathrm{CC\text{-}PC\text{-}TE\text{-}NUW}
	$.
\end{theorem}
\begin{proof}
    $\mathrm{CC\text{-}PC\text{-}TE\text{-}UW}
	\subseteq
	\mathrm{CC\text{-}PC\text{-}TE\text{-}NUW}$ follows from definitions.
    Let $(C, V, p)$ be a element of CC-RPC-TE-UW via partition $(C_1, C_2)$, and let $E=(C, V)$\@.
    Without loss of generality, let $p \in C_2$. Due to the TE tie-handling rule, $p$ must be a  unique winner of $(C_2,V)$.
	
	If there is no winner under $C_1$, then under $\mathrm{CC\text{-}PC\text{-}TE\text{-}UW}$, we can directly use the partition $(C_1,C_2)$. In this case, no candidate advances from the first stage, the final candidate set is $C_2$, and $p$ is a unique winner in it.
	
	If there is a winner under $C_1$, there can only be one survivor, denoted as $q$, and $p$ defeats $q$ in pairwise contest.

	Now define 
    $R=\{\,r\in C_2\setminus\{p\}\mid q\succcurlyeq_E r \text{ and } r\succ_E p\,\}$, 
    and consider the partition 
	$(D_1,D_2)=(C_1\cup R,\;C_2\setminus R)$.
	
	Since TE leaves at most one survivor in the first phase, we consider the following cases. %
	
	\smallskip
	\noindent\textbf{Case~1: There are no survivors in \boldmath$D_1$. } %
	At this time, the final candidate set is $D_2=C_2\setminus R$. Since every candidate in $R$ defeats $p$, by Lemma~\ref{lem:llull-remove-defeaters-preserve-winner} and the fact that $p$ is a unique winner in $C_2$, after removing $R$, $p$ remains a unique winner.
	
	\smallskip
	\noindent\textbf{Case~2: The unique survivor in \boldmath$D_1$ is $q$. } 
	Now, the final candidate set is $(C_2\setminus R)\cup\{q\}$.	
	For any $x\in C_2\setminus(R\cup\{p\})$, $x$ provides at least as many points to $p$ as it provides to $q$; otherwise, it would contradict the definition of $R$. Furthermore, since $p$ pairwise defeats $q$, $p$ strictly surpasses $q$ in the final round.
	
	On the other hand, after removing $R$, $p$ is still a unique winner of $(C_2\setminus R, V)$. Adding $q$ back, $p$ gains exactly one more point by defeating $q$, and any $d\in C_2\setminus(R\cup\{p\})$ gains also at most one more point, so $p$ still strictly surpasses every such $d$.	
	Therefore, $p$ remains a unique winner in the final round.
	
	\smallskip
	\noindent\textbf{\boldmath Case~3: The unique survivor of $D_1$ is some $r\in R$. }
	At this point, the final candidate set is
	$(C_2\setminus R)\cup\{r\}
	= C_2\setminus(R\setminus\{r\})$.
	Since each candidate in $R\setminus\{r\}$ defeats $p$, 
    leveraging Lemma~\ref{lem:llull-remove-defeaters-preserve-winner}
    allows us to 
    conclude that $p$ remains a unique winner in the final round.
	
	\smallskip	
	\noindent\textbf{\boldmath Case~4: The unique survivor in $D_1$ is some $y\in C_1\setminus\{q\}$. }
	This situation is impossible. Because $q$ is a unique winner in $C_1$, it has already strictly outperformed $y$ in $C_1$.
	After adding $R$, according to the definition of $R$, $q$ will not perform worse than $y$ against these new opponents, therefore $y$
	cannot 
    become the unique survivor.

	In summary, the instance $(C, V, p)$ is in $\mathrm{CC\text{-}PC\text{-}TE\text{-}UW}$ via the candidate partition
	$
	(C_1\cup R,\;C_2\setminus R)
	$.~\qed
\end{proof}

It follows from the above proof that we can obtain an additional containment result.

\begin{corollary}\label{cor:llull-cc-rpc-te-nuw-subset-pc-te-nuw}
	Under Llull, $\mathrm{CC\text{-}RPC\text{-}TE\text{-}NUW}
	\subseteq \mathrm{CC\text{-}PC\text{-}TE\text{-}NUW}$.
\end{corollary}
\begin{proof} 
	Following the construction and classification discussion of Theorem~\ref{thm:llull-cc-rpc-te-uw-subset-pc-te-uw}, the only modification needed is for Case~2 as the other cases are identical. Let $(C, V, p)$ be in $\mathrm{CC\text{-}RPC\text{-}TE\text{-}NUW}$ via partition $(C_1, C_2)$, and define $R$ and $q$ as in the proof of Theorem~\ref{thm:llull-cc-rpc-te-uw-subset-pc-te-uw}. We only need to show how to adapt Case~2 in this proof.

	In Case~2, for each $x\in C_2\setminus (R \cup \{p\})$, $x$ provides at least as many points to $p$ as it provides to $q$; otherwise, it would contradict the definition of $R$. Furthermore, since $p$ is one of the winners in the  $\mathrm{CC\text{-}PC\text{-}TE\text{-}NUW}$
	final $(\{p,q\},V)$, under Llull, $p$'s total score in the final is at least as high as $q$. On the other hand, as the previous proof showed, $p$ still strictly exceeds every $d\in C_2\setminus(R\cup\{p\})$. Therefore, $p$ is at least one of the winners in the final round.

    In summary, the instance $(C, V, p)$ is in $\mathrm{CC\text{-}PC\text{-}TE\text{-}NUW}$ via the candidate partition
    $(C_1\cup R,\;C_2\setminus R)$.~\qed
\end{proof}

This next proof builds on the proof of Theorem~\ref{thm:llull-cc-rpc-te-uw-subset-pc-te-uw} and references it. However, it is significantly more involved as simply moving around the candidates from $R$ is not sufficient in this next result. We thus proceed in a two-phase fashion.

\begin{theorem}\label{thm:llull-cc-rpc-tp-nuw-subset-pc-te-nuw}
    Under Llull,
	$\mathrm{CC\text{-}RPC\text{-}TP\text{-}UW} \subseteq \mathrm{CC\text{-}RPC\text{-}TP\text{-}NUW} \subseteq \mathrm{CC\text{-}PC\text{-}TE\text{-}NUW}$.
\end{theorem}
\begin{proof}
    Under every voting rule, $\mathrm{CC\text{-}RPC\text{-}TP\text{-}UW}
	\subseteq
	\mathrm{CC\text{-}RPC\text{-}TP\text{-}NUW}$ holds by definition, so it suffices to show that under Llull,
    $\mathrm{CC\text{-}RPC\text{-}TP\text{-}NUW}
	\subseteq
	\mathrm{CC\text{-}PC\text{-}TE\text{-}NUW}$.

    Let $(C, V, p)$ be an element of $\mathrm{CC\text{-}RPC\text{-}TP\text{-}NUW}$ via partition $(C_1, C_2)$.
    Without loss of generality, let $p\in C_2$. Let $D_1$ and $D_2$ be the respective winner sets of $(C_1,V)$ and $(C_2,V)$ under Llull\@. Then $p\in D_2$, and $p$ is one of the winners of the final round $(D_1\cup D_2,V)$.
	
	If $\card{D_1}\neq 1$, then the original partition can be used directly under CC-PC-TE-NUW,  because in the TE model, if there is not a unique winner in a subelection, no one from that subelection advances to the final round. Therefore, the candidate set for the final is $C_2$. Since $p\in D_2$, $p$ is one of the winners of $(C_2,V)$. Therefore, we only need to consider the case where $D_1=\{q\}$, where $q$ is a unique winner of $(C_1,V)$.
	
	The new partition will be constructed in two phases. Initialize $C_1' = C_1$ and $C_2'=C_2$. The two sets will be updated procedurally.

    \noindent
    \textbf{Start of Phase~1.}
    As long as there exists some candidate $r \in C_2'$ that satisfies
	$q\succcurlyeq r\ \text{and}\ r\succ p$, 
    then move $r$ to $C_1^{\prime}$. By Lemma~\ref{lem:llull-remove-defeaters-preserve-winner}, each such move will not disrupt $p$'s winning status in $C_2^{\prime}$.

    After this process, we consider a few cases, which are 
    handled in exactly the same way as the proof of Theorem~\ref{thm:llull-cc-rpc-te-uw-subset-pc-te-uw}.	
	If there is no unique winner in $(C_1',V)$, or if the newly added candidate $r$ becomes a unique winner in $(C_1',V)$, then a successful  control action is immediately obtained.	
	Furthermore, no other candidate in $C_1 \setminus \{q\}$, can become a unique winner of $(C_1', V)$.
	Since the number of candidates is finite, the above process will inevitably terminate.	
	Therefore, if it terminates without a successful partition, then by construction, we obtain a partition $(C_1',C_2')$ such that
	(i)~$q$ is a unique winner of $(C_1',V)$,
	(ii)~$p$ is a winner of $(C_2',V)$, and
	(iii)~for each $x\in C_2'\setminus\{p\}$, the contribution of $x$ to $q$'s score in $(\{q\} \cup C_2', V)$---which is the final-round election under PC-TE---does not exceed the contribution of $x$ to $p$'s score in the same election.
	
	From~(iii), we know that no candidate in $C_2'$ is more ``biased'' towards $q$ than towards $p$. Since $q$ did not strictly exceed $p$ in the original $\mathrm{CC\text{-}RPC\text{-}TP\text{-}NUW}$ final $(\{q\}\cup D_2,V)$, it is still impossible for $q$ to strictly exceed $p$ in the current final $(\{q\}\cup C_2',V)$. The candidates removed from $D_2$ all satisfy $q\succcurlyeq r$ and $r\succ p$, so removing them only weakens the advantage of $q$ relative to $p$ and the newly added candidates, as shown in~(iii), will not enhance the advantage of $q$ relative to $p$. Therefore, if the current partition is still unsuccessful, 
    it can only be due to
    the other candidates in $C_2'$.
    
    \noindent
    \textbf{End of Phase~1.}
    
    \smallskip
    
    \noindent
    \textbf{Start of Phase~2.}
    We call $d\in C_2'$ a dangerous candidate if (a)~$d$ and $p$ are winners of $(C_2',V)$, (b)~$q\succ p$, and (c)~$d\succcurlyeq q$. That is, $d$ ties with $p$ in $(C_2', V)$ first, then surpasses $p$ in the final round thanks to $q$'s $1$ point. If there are no dangerous candidates, the construction is complete. Otherwise, fix a dangerous candidate $d$. 

    We first establish that $d \in D_2$. 
    This is because we only removed candidates who defeated $p$ in Phase~1, $p$'s score in $(C_2', V)$ is the same as in $(C_2, V)$ and removing those candidates did not increase anyone's score. Since $d$ is still a winner with $p$ in $(C_2',V)$, it must also be tied with $p$ in $(C_2,V)$.
	By the same token,
    each previously removed candidate beats $d$. This is because $d$ was tied with $p$ in $(C_2, V)$ and remains tied with $p$ in $(C_2', V)$, and since $p$ never lost points during the removal process, $d$ could not have lost points either.

	Let $D=D_2\cap C_2'$, i.e., the original set of winners of $(C_2,V)$ that have not yet been removed from $C_2'$. We claim that in $(D,V)$, $p$ must strictly lead $d$.
	Otherwise, $p$ will at most tie with $d$ in $(D,V)$.
	Since $q\succ p$ and $d\succcurlyeq q$, adding $q$ would give $d$ one more point than $p$, thus strictly surpassing $p$ in $(\{q\}\cup D,V)$.
	This is impossible because the candidates removed from $D_2$ all beat both $p$ and $d$, removing them does not change the score difference between $p$ and $d$.
	Therefore, if $d$ strictly exceeds $p$ in $(\{q\}\cup D,V)$, it will also strictly exceed $p$ in the original $\mathrm{CC\text{-}RPC\text{-}TP\text{-}NUW}$ final $(\{q\}\cup D_2,V)$, contradicting the fact that $p$ is one of the winners.
	Therefore, $p$ must strictly lead $d$ in $(D,V)$, and since $p$ and $d$ tie in $(C_2', V)$, there must exist some $x\in C_2'\setminus D$ such that
	$
	d\succcurlyeq x$ and $x\succ p$.
    Furthermore, since $x\succ p$ and $x$ was not removed in Phase~1, it cannot be the case that $q\succcurlyeq x$. Therefore, $x\succ q$.
	
	Now, move this $x$ from $C_2'$ to $C_1'$. As before, if there is no unique winner in $(C_1', V)$ or the newly moved $x$ becomes the unique winner in $(C_1', V)$, then we have obtained a desired partition, because $q$ initially leads all the old candidates by at least $1$ point, and adding a candidate $x$ can only reduce each gap by a maximum of $1$ point, so no old candidate can strictly surpass $q$. 
    The worst case is that an old candidate ties with $q$, which is the previously determined ``no unique winner'' situation. Otherwise, $q$ remains a unique winner in $(C_1', V)$, and the original dangerous candidate $d$ is no longer dangerous because it originally caught up with $p$ precisely because of the $1$ point from $x$, and now that $x$ has been moved, $d$ will lose that 1 point, while $p$'s score remains unchanged.
	
	Note that this step does not create a new dangerous candidate. This is because we are still removing a candidate that defeats $p$. Furthermore, since $x\succ q$ and $x\succ p$, removing $x$ does not enhance $q$'s advantage over $p$.
    
	Repeat the above process as long as a dangerous candidate exists and a desired partition has not been constructed.
	Each iteration of the process either successfully constructs a desired partition or eliminates at least one dangerous candidate without creating a new one. Since the number of dangerous candidates is finite, this iterative process terminates. Termination occurs either because a previous step has succeeded, or because a unique winner of $(C_1', V)$ remains $q$, but there are no more dangerous candidates in $C_2'$. In the latter case, together with the fact already proved after the first phase that $q$ cannot strictly beat $p$, this implies that $p$ remains a winner of the final round. 

    \noindent
    \textbf{End of Phase~2.}
	
	In summary, starting from any successful $\mathrm{CC\text{-}RPC\text{-}TP\text{-}NUW}$ partition, we can construct a successful $\mathrm{CC\text{-}PC\text{-}TE\text{-}NUW}$ partition.\qed

\end{proof}

This next result also follows a procedural process with iterative steps, but we consider far more cases in our proof of Theorem~\ref{thm:llull-cc-pc-tp-uw-subset-pc-te-nuw} than we did in the proof of Theorem~\ref{thm:llull-cc-rpc-tp-nuw-subset-pc-te-nuw}. We will define a ``stopping rule,'' i.e., a condition that will drive our primary iterative process.

\begin{theorem}\label{thm:llull-cc-pc-tp-uw-subset-pc-te-nuw}
	Under Llull, $\mathrm{CC\text{-}PC\text{-}TP\text{-}UW}
	\subseteq
	\mathrm{CC\text{-}PC\text{-}TE\text{-}NUW}.$
\end{theorem}
\begin{proof}
    Let $(C, V, p)$ be an element of CC-PC-TP-UW via partition $(C_1, C_2)$.
	Let $D_1$ be the set of winners of $(C_1, V)$. Under TP, the final-round candidates are $C_2\cup D_1$, and by assumption $p$ is a unique winner of that round.
	If $\card{D_1}\leq 1$, then TP and TE behave identically on subelection $(C_1, V)$, so the original partition already works under CC-PC-TE-NUW\@. Henceforth assume $\card{D_1}>1$.
	
	We first separate the two possible locations of $p$.
	If $p\in C_2$ and $p$ is already a winner in $(C_2, V)$, then the original partition works under TE.
	Since $\card{D_1}>1$, no candidate from $C_1$ advances, and the final is $(C_2, V)$.
	
	Thus we may assume that either $p\in C_1$, or $p\in C_2$ but $p$ is not a winner of $(C_2, V)$.
	If $p\in C_1$, then necessarily $p\in D_1$, since otherwise $p$ would not appear in the TP final round.
	
	Consider
	$C_1'=C_1\setminus D_1$ and $C_2'=C_2\cup D_1$.
	Then $p$ is the unique winner of $(C_2', V)$.
	Throughout the following construction, whenever a candidate is moved from one part to the other,
	we update $C_1'$ and $C_2'$ accordingly, so that they always denote the current partition.
	If $(C_1',V)$ has no unique winner, then under TE no candidate advances, and the final is $C_2'$. 
	Hence $p$ wins.
	Thus assume $C_1'$ also has unique winner $q$, who is naturally not in $D_1$.
	
	We will repeatedly use the following stopping rule. Whenever we move a candidate
	$r$ with $r\succ_M p$ from the current $C_2'$ to the current $C_1'$, we
	check the winner set of $C_1'$. If $C_1'$ has no unique winner, then
	TE eliminates it, and the final is the current $C_2'$. If the newly moved
	candidate $r$ becomes the unique winner, then $r$ advances back to the final. In both
	cases, the final is obtained from the original safe set $C_2\cup D_1$ by deleting
	only candidates that defeat $p$. Hence, by
	Lemma~\ref{lem:llull-remove-defeaters-preserve-winner}, $p$ remains a winner.

	First move candidates $r\in C_2'$ such that
	$r\succ p$ and $r\succ q$, one at a time. After each move, apply the stopping rule above.
	Moreover, no old candidate that already belonged to the current $C_1'$ before this move, except for $q$, can become a new unique winner before first tying $q$.
	Since the newly moved candidate defeats $q$, $q$ receives no point from it, while any candidate that was already in $C_1'$ receives at most one point from it. 
	Thus one move can reduce the score gap between $q$ and any old candidate by at most one. 
	As $q$ was a unique winner before the move, such a candidate can at best tie $q$, but cannot strictly surpass $q$. 
	Thus that candidate cannot become a new unique winner at this step. If they ties $q$, we are in the already successful no-unique-winner case.	
	
	If $q$ is still the unique winner after all such candidates have been moved, move each remaining candidate $r\in C_2'$ that defeats $p$, again one at a time.
	After each move, apply the same stopping rule. Since all candidates that defeat both
	$p$ and $q$ have already been moved, every remaining moved candidate $r$ satisfies
	$q\succcurlyeq_M r$. Thus $q$ receives a point from $r$, while any candidate already
	in the previous $C_1'$ before the move receives at most one point. Hence no old candidate can improve
	relative to $q$ during this second moving step. Therefore, if the construction has
	not already succeeded, then after all candidates defeating $p$ have been moved, the
	first part still has unique winner $q$.
	
	At this point, no candidate in $C_2'\setminus\{p\}$ defeats $p$. Since
	$C_2'$ is obtained from $C_2\cup D_1$ by deleting candidates that defeat
	$p$, Lemma~\ref{lem:llull-remove-defeaters-preserve-winner} gives that $p$ is a unique winner in $(C_2', V)$.
	
	If $q$ does not defeat $p$, then $p$ is a weak Condorcet winner of the final round $(C_2'\cup\{q\}, V)$. Hence $p$ is still a winner, and the current partition is successful.
	
	If there is some $s\in C_2'\setminus\{p\}$ such that
	$p\succcurlyeq s$ and $s\succ q$, then the point that $q$ obtains from defeating
	$p$ is compensated by $s$: $q$ receives no point from $s$, while $p$ receives a point
	from $s$. Thus $q$ cannot strictly exceed $p$ in $(C_2'\cup\{q\}, V)$. Moreover, every
	candidate in $C_2'\setminus\{p\}$ was strictly behind $p$ in $(C_2', V)$ and can gain at most
	one point from $q$, so it can at best tie $p$. Since we are in the NUW model,
	$p$ is a winner.
	
	It now remains to consider the case where no such candidate $s$ exists. We return to the
	initial partition $(C_1\setminus D_1,\, C_2\cup D_1)$ and handle the old winners in
	$D_1$\@.
	Let $E=\{d\in D_1\setminus\{p\}\mid d\succ p\}$.

	First suppose $E\neq\emptyset$. Starting again from the partition
	$(C_1\setminus D_1,\, C_2\cup D_1)$, move the candidates in $E$ from the second part
	to the first part, one at a time, and apply the stopping rule after each move.
	
	We claim that this process must eventually trigger the stopping rule. Fix some $d\in E$.
	Since $d\in D_1$ and $q\notin D_1$, we have
	$\score_{(C_1, V)}(d)>\score_{(C_1, V)}(q)$. If all candidates in $E$ are moved back, then the
	left side is obtained from the original $C_1$ by deleting exactly the candidates in
	$D_1\setminus E$. For each candidate
	$x\in D_1\setminus(E\cup\{p\})$, we have $p\succcurlyeq x$ by the definition of $E$,
	and, since there is no candidate $s$ where $p\succcurlyeq s$ and $s\succ q$, it also holds that
	$q\succcurlyeq x$. Thus deleting such an $x$ removes one point from $q$ and at most
	one point from $d$. If $p$ is also deleted, then, since $q\succ p$, deleting $p$ again
	removes one point from $q$ and at most one point from $d$. Hence deleting candidates in
	$D_1\setminus E$ does not improve $q$'s score relative to $d$'s score. Therefore $d$ still strictly
	beats $q$ scorewise after all candidates in $E$ have been moved back. So $q$ cannot remain
	a unique winner throughout the whole process. At the first step where $q$ ceases to be
	a unique winner, the stopping rule gives a successful partition.

	Now suppose $E=\emptyset$. Then no candidate in $D_1\setminus\{p\}$ defeats $p$.
    Recall also that there is 
    no candidate $s$ such that $p\succcurlyeq s$ and $s\succ q$, so 
    every
	candidate in $D_1\setminus\{p\}$ is tied-or-defeated by both $p$ and $q$.
	As $\card{D_1}>1$, choose some $d\in D_1\setminus\{p\}$.
	
	Define
	$\widehat C_1=(C_1\setminus D_1)\cup\{d\}$ and
	$\widehat C_2=(C_2\cup D_1)\setminus\{d\}$.
	We show that $d$ is a unique winner of $(\widehat C_1,V)$.
	
	Since $d\in D_1$ and $q\notin D_1$, we have
	$\score_{(C_1, V)}(d)>\score_{(C_1, V)}(q)$. Going from $C_1$ to $\widehat C_1$, we delete the
	candidates in $D_1\setminus\{d\}$. Each deleted candidate
	$x\in D_1\setminus\{d,p\}$ is tied-or-defeated by $q$, so deleting $x$ removes one point
	from $q$ and at most one point from $d$. If $p$ is deleted, then $q\succ p$, so
	deleting $p$ again removes one point from $q$ and at most one point from $d$.
	Thus deleting $D_1\setminus\{d\}$ does not improve $q$'s score relative to $d$'s score, and hence
	$\score_{(\widehat C_1, V)}(d)>\score_{(\widehat C_1, V)}(q)$.
	
	On the other hand, $q$ was a unique winner of $(C_1\setminus D_1,V)$. Adding back
	the single candidate $d$ gives $q$ one point, while every other candidate already in
	the first part receives at most one point. Hence no such candidate can overtake $q$.
	
	Therefore, $d$ is a unique winner of the first-round election, and the final-round candidates are
	$\widehat C_2\cup\{d\}=C_2\cup D_1$. Therefore this partition is successful.
	
	In all cases, we obtain a successful
	$\mathrm{CC\text{-}PC\text{-}TE\text{-}NUW}$ partition for $p$. Therefore,
	$
	Llull\mathrm{\text{-}CC\text{-}PC\text{-}TP\text{-}UW}
	\subseteq
	Llull\mathrm{\text{-}CC\text{-}PC\text{-}TE\text{-}NUW}.
	$\qed
\end{proof}

These are the only new containments we found under Llull, and we give our separations in Table~\ref{tab:llull-relationships}. We leave open the relationships between three pairs. In particular, under Lllull, it is not known whether CC-PC-TP-NUW is contained in CC-PC-TE-NUW, and it is not known whether DC-PV-TP-NUW is contained in both DC-PV-TE-UW and DC-PV-TE-NUW\@.

\subsection{Schulze}\label{subsec:schulze}

Schulze is also a voting rule based on pairwise comparisons of candidates. Unlike Copeland, which only considers the direct pairwise result, Schulze compares the ``strength'' of the ``strongest paths'' from one candidate to another.

Let $E=(C, V)$ be a Schulze election and consider the weighted majority graph of $E$ with the restriction that all edges have positive edges weights.
Now, an $a \to b$ path in that graph is a sequence of the form $a \to x_1 \to \cdots \to x_{t-1} \to b$, where each candidate is pairwise distinct and every edge in the path has a positive pairwise margin. The strength of such a path is the minimum edge weight on it.

Let $P_E(a,b)$ denote the strength of the strongest $a \to b$ path, which is the maximum path strength from $a$ to $b$. If no such path exists, set $P_E(a,b)=0$. 

When $P_E(a,b) \geq P_E(b,a)$, we say $a$ beats/defeats $b$, denoted as $a\succ_{S,E} b$. A Schulze winner is a candidate who is not defeated by any other candidate.
Under Schulze, if the candidate set is nonempty, then the winner set is nonempty~\cite{sch:j:schulze}.

\begin{example}\label{ex:schulze-basic}
	Let $C=\{a,b,c,d\}$, and consider the following profile:
	\[
	\begin{array}{rll}	
		3\text{ voters:} & a > c > d > b\\	
		2\text{ voters:} & b > a > c > d\\	
		2\text{ voters:} & c > d > b > a\\	
		1\text{ voters:} & d > b > a > c		
	\end{array}	
	\]
	The pairwise margins are:	
	$
	m(a,b)=-2, m(a,c)=4, m(a,d)=2,
	m(b,c)=-2, m(b,d)=-4, m(c,d)=6$.	
	The corresponding weighted majority plot is shown in Figure~\ref{fig:schulze-wmg-example}.

	\begin{figure}[ht]
		\centering
        \caption{The weighted majority graph in Example~\ref{ex:schulze-basic}.}
		\label{fig:schulze-wmg-example}
		\begin{tikzpicture}[
            scale=0.6,
			>=Latex,
			vertex/.style={circle, draw, line width=0.6pt, minimum size=8mm, inner sep=0pt},
			edge/.style={->, semithick}
			]
			\node[vertex] (a) at (0,1.4) {$a$};
			\node[vertex] (b) at (4.6,1.4) {$b$};
			\node[vertex] (c) at (1.5,3.2) {$c$};
			\node[vertex] (d) at (2.9,-0.3) {$d$};
			
			\draw[edge] (a) -- (c)
			node[pos=.55, above left=2pt, fill=white, inner sep=1pt] {$4$};
			\draw[edge] (a) -- (d)
			node[pos=.30, below left=2pt, fill=white, inner sep=1pt] {$2$};
			\draw[edge] (c) -- (b)
			node[pos=.78, above right=1pt, fill=white, inner sep=1pt] {$2$};
			\draw[edge] (d) -- (b)
			node[pos=.55, below right=3pt, fill=white, inner sep=1pt] {$4$};
			\draw[edge] (c) -- (d)
			node[pos=.78, right=6pt, fill=white, inner sep=1pt] {$6$};
			\draw[edge] (b) to[bend right=18]
			node[pos=.82, below=9pt, fill=white, inner sep=1pt] {$2$} (a);
		\end{tikzpicture}
	\end{figure}
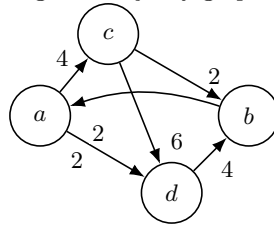

Although $a$ loses to $b$ in direct pairwise contest, when inspecting all the $a \to b$ paths, i.e., 
$a\to c\to b, a\to d\to b, a\to c\to d\to b$,
we see their strengths are respectively, 
$\min\{4,2\}=2, \min\{2,4\}=2, \min\{4,6,4\}=4$.
Therefore, the strongest path strength from $a$ to $b$, $P(a,b)$, is 4.
In the opposite direction, there is only one path $b\to a$ from $b$ to $a$, so $P(b,a)=2$.
Since $P(a,b) \geq P(b,a)$, $a$ defeats $b$ under Schulze.
	
Similarly, $P(a,c)=4 \geq 2=P(c,a)$
(using $a\to c$ and $c\to d\to b\to a$), and $P(a,d)=4 \geq 2=P(d,a)$
(using $a\to c\to d$ and $d\to b\to a$).

Thus, $a$ also beats $c$ and $d$ under Schulze and is a unique winner.
\end{example}

A Condorcet-consistent voting rule elects a Condorcet winner when one exists, and a weak-Condorcet-consistent voting rule elects a weak Condorcet winner---that is, one that ties or defeats every other candidate in a head-to-head contest--when one exists.

\begin{observation}[\protect\cite{sch:j:schulze}]\label{obs:schulze-consistency}
	Schulze satisfies both weak Condorcet consistency and Condorcet consistency. More concretely, 
	\begin{enumerate}
		\item If $p$ is a weak Condorcet winner of election $(C, V)$, then $p$ will always be a Schulze winner of $E$.
		\item If $p$ is a Condorcet winner of election $(C, V)$, then $p$ will always be a unique Schulze winner of $E$.
	\end{enumerate}
\end{observation}

We first address the Group~4 relationships.

\begin{proposition}\label{prop:schulze-group4}
    Schulze
    satisfies the majority criterion and fails to satisfy Properties~$\alpha$ and Unique-$\alpha$ via a 1-abstention-safe election. 
\end{proposition}
\begin{proof}
    The satisfaction of the majority criterion is immediate as a majority winner is also a Condorcet winner, and thus is always a unique winner under Schulze.

    Let $C=\{a,b,c,d\}$ and consider the weighted majority graph with
    positive weights for the following edges
    \[
    a \to c : 12,\quad c \to b : 10,\quad a \to d : 8,\quad d \to b : 6,\quad b \to a : 4,\quad c \to d: 2,
    \]
    and no other positive edge weights.

    Let $(C, V)$ denote the election constructed by McGarvey's Theorem using the aforementioned weighted majority graph.
    In  election on $(C, V)$, the strongest path from $a$ to $b$ is
    $a \to c \to b$, of strength $\min\{12,10\}=10$, while the strongest path from $b$ to $a$
    has strength $4$ (via the direct edge $b \to a$). Thus $a$ beats $b$ in the Schulze
    sense. Also, $a$ beats both $c$ and $d$ directly, so $a$ is a unique Schulze winner in $(C, V)$.
    In the subelection $(C', V)$, where $C'=\{a,b\}$,  $b$ pairwise defeats $a$, and hence $b$ is
    a unique Schulze winner. Thus, $a$ is a unique winner in $(C, V)$ but not in $(C', V)$, and Schulze fails to satisfy Properties~$\alpha$ and Unique-$\alpha$. Moreover, notice that deleting a single voter $v'$ can only decrease each weight in the weighted majority graph by at most 1, so $a$ remains unique winner in $(C, V-\{v'\})$.
    \qed
\end{proof}

We now focus on Group~5 relationships and begin by introducing some useful lemmas that will be used to prove our collapses under Schulze\@.

\begin{lemma}\label{lem:schulze-co-winner-must-tie-p}	
Let $E=(C, V)$ be a Schulze election with at least one candidate $p$.
Let $S\subseteq C$ and let $p\in S$. If $p$ is a weak Condorcet winner in the subelection $(S,V)$, and some candidate $r\neq p$ is also a Schulze winner in $(S,V)$, then $r$ and $p$ tie in a head-to-head contest.	
\end{lemma}

\begin{proof}	
	Since $p$ is a weak Condorcet winner of $(S,V)$, $p$ does not lose in a head-to-head contest against any candidate. Therefore, no path to $p$ can have positive strength, and the strongest path from $r$ to $p$ must have strength 0.
	If $p$ strictly defeats $r$ in a head-to-head contest, then there is at least one path from $p$ to $r$ that has positive strength, and that path is a single edge. Therefore, the strongest path strength from $p$ to $r$ is greater than that from $r$ to $p$, meaning $p$ defeats $r$ under Schulze. This contradicts the fact that $r$ is also a Schulze winner. Therefore, $r$ cannot lose to $p$. Thus, $p$ and $r$ are tied.\qed
\end{proof}

\begin{lemma}\label{lem:schulze-all-finalists-tie-p}	
Let $E=(C, V)$ be a Schulze election with at least one candidate $p$.
Let $F\subseteq C$, and assume that both $p\in F$ and $\card{F}\ge 2$. If each candidate in $F\setminus\{p\}$ is ties with $p$ in a head-to-head contest, then $p$ is not a unique Schulze winner of $(F,V)$.	
\end{lemma}
\begin{proof}
	Since $p$ is tied with each candidate in $F\setminus\{p\}$, for any $x\in F\setminus\{p\}$, the strongest path strength from $p$ to $x$ and from $x$ to $p$ is $0$. Therefore, $p$ does not strictly defeat any other  candidate under Schulze.
	In the election $(F\setminus\{p\},V)$, there is at least one Schulze winner (since $F\setminus\{p\}$ is nonempty by assumption), denoted by $y$. After adding $p$ back to the election, since the pairwise margin of $p$ and each other candidate are $0$, the path through $p$ does not increase the path strength between the other candidates, so $y$ remains a winner.
	On the other hand, because $p$ is tied with each other candidate in $F\setminus\{p\}$, $p$ is a weak Condorcet winner in $(F,V)$. By Observation~\ref{obs:schulze-consistency}, $p$ is also a Schulze winner.
	Therefore, $(F,V)$ has at least two Schulze winners: $p$ and $y$. Thus, $p$ is not a unique winner.	\qed
\end{proof}

Equipped with our two lemmas, we now present eight new collapsing pairs under Schulze\@.

\begin{theorem}\label{thm:dc-te-and-tp-uw-collapse-schulze}
    Under Schulze,
	$
	\mathrm{DC\text{-}PC\text{-}TE\text{-}UW} = 
	\mathrm{DC\text{-}PC\text{-}TE\text{-}NUW} = 
	\mathrm{DC\text{-}RPC\text{-}}\allowbreak \mathrm{TE\text{-}UW} = 
	\mathrm{DC\text{-}RPC\text{-}TE\text{-}NUW} =	
	\mathrm{DC\text{-}PC\text{-}TP\text{-}UW} = 
	\mathrm{DC\text{-}RPC\text{-}TP\text{-}UW}$.	
\end{theorem}
\begin{proof}
	We prove that membership in these six control types
    is characterized by the same condition: $(C, V, p)$ is a Yes-instance 
    if and only if $p$ is not a Condorcet winner of $(C ,V)$\@.
	Menton and Singh~\cite{men-sin:c:schulze} already gave the characterization for the $TE\text{-}NUW$ cases, i.e., 
	$\mathrm{DC\text{-}RPC\text{-}TE\text{-}NUW}$ and	
	$\mathrm{DC\text{-}PC\text{-}TE\text{-}NUW}$\@.	
    Furthermore, the universal collapses of Hemaspaandra, Hemaspaandra, and Menton~\cite{hem-hem-men:j:search-versus-decision} imply that the characterization holds for the four types that use the TE tie-handling rule.
    So it remains to show the characterization for the remaining two control types, which use the TP tie-handling rule.
	
	If $p$ is a Condorcet winner of $(C, V)$, then in any subelection containing $p$,
	$p$ remains a Condorcet winner because removing other candidates does not change the pairwise results among the remaining candidates.
	By Observation~\ref{obs:schulze-consistency}, $p$ will be a unique Schulze winner in every such subelection.
	Therefore, regardless of whether the candidate-partitioning method is PC or RPC, $p$ will always reach the final stage and remain a unique winner.
	So $(C, V, p)$ is a No-instance of both
	$\mathrm{DC\text{-}PC\text{-}TP\text{-}UW}$ and
	$\mathrm{DC\text{-}RPC\text{-}TP\text{-}UW}$.
	
	Conversely, suppose $p$ is not a Condorcet winner of $(C, V)$\@.
	First consider, if $p$ is not even a weak Condorcet winner, then there must exist a candidate $q\neq p$,	
	such that $q$ defeats $p$.	
	Let $C_1=\{q,p\}, C_2=C\setminus\{q,p\}$,	
	this partition will eliminate $p$ in the first round under both candidate-partitioning methods, i.e., RPC and PC, and thus $(C, V, p)$ is a Yes-instance of both $\mathrm{DC\text{-}PC\text{-}TP\text{-}UW}$ and
	$\mathrm{DC\text{-}RPC\text{-}TP\text{-}UW}$.
    Consider now the case where $p$ is a weak Condorcet winner, but not a Condorcet winner. Then there exists a candidate $q\neq p$ such that $q$ and $p$ are tied in a head-to-head contest.	
	Let $C_1=C\setminus\{q\}, C_2=\{q\}$.	
    Under both candidate-partitioning methods (i.e., RPC and PC), the candidates in the final round will be the Schulze winners of $(C_1, V)$ and $q$.
    Since removing a candidate does not change the pairwise results among the remaining candidates, $p$ remains a weak Condorcet winner in $(C_1,V)$.	
	By Observation~\ref{obs:schulze-consistency}, $p$ is a Schulze winner in $(C_1,V)$.
	Let $D_1$ be the Schulze winners of $(C_1,V)$.	
	By Lemma~\ref{lem:schulze-co-winner-must-tie-p}, each candidate in $D_1\setminus\{p\}$ is tied with $p$ in a head-to-head contest.
	Furthermore, by construction, $q$ also has a tie with $p$.
	Therefore, in the final-round candidate set
	$D_1\cup\{q\}$, every candidate ties with $p$.
	By Lemma~\ref{lem:schulze-all-finalists-tie-p},	
	$p$ is not a unique Schulze winner in the final round.
	Therefore, $(C, V, p)$ is a Yes-instance of btoh
	$\mathrm{DC\text{-}PC\text{-}TP\text{-}UW}$ and
	$\mathrm{DC\text{-}RPC\text{-}TP\text{-}UW}$.\qed
\end{proof}

These are the only new collapses under Schulze, and we give our separations in Table~\ref{tab:schulze-relationships}.

\subsection{Ranked Pairs}\label{subsec:ranked-pairs}

Ranked Pairs (RP), also known as the Tideman method~\cite{tid:j:clones-dodgson}, is a Condorcet-consistent---that is, it always elects a Condorcet winner if one exists---rule based on pairwise majority margins between candidates. Unlike Schulze, RP first sorts the margins from largest to smallest, then iterates over them one by one doing the following: if adding an edge from $a$ to $b$ creates a directed cycle, proceed to the next edge, otherwise add the edge $(a,b)$ to the graph. As a result, a directed, acyclic relation is obtained, thus providing a ranking of candidates. Under RP, the candidate ranked first after that process is a unique winner.
One common issue when defining RP involves breaking ties 
when iterating over candidates with equal majority margins.
One approach to handle this, which was suggested by Brill and Fischer~\cite{bri-fis:c:neutrality-ranked-pairs}, is to fix a global (polynomial-time computable) tie-breaking rule. 
For concreteness, we use the one used by Hemaspaandra, Lavaee, and Menton~\cite{hem-lav-men:j:schulze-and-ranked-pairs} who mention that it is the tie-breaking rule later given by Tideman~\cite{tid:b:choice} in his later book; we refer interested readers to Footnote~1 of \cite{hem-lav-men:j:schulze-and-ranked-pairs}. In particular,
to break ties between two edges (which are pairs) $(a,b)$ and $(c, d)$ with equal majority margin,  
we favor ``the pair with the lexicographically-larger larger-candidate-of-the-pair, and when the larger members
are the same in both pairs then breaking the tie in favor of whichever pair has the lexicographicallylarger smaller-candidate-of-the-pair.''~\cite{hem-lav-men:j:schulze-and-ranked-pairs}

For example, if $C = \{a,b,c\}$ and the margins satisfy $N(a,b) = 4$, $N(b,c) = 2$, and $N(c,a) = 1$, then RP will first add edge $(a, b)$, followed by edge $(b, c)$. Edge $(c, a)$ will not be added as doing so would create cycle. The resulting directed, acyclic graph then forms a ranking over the candidates that ranks $a$ first, making $a$ a unique winner under RP.

We first address the Group~4 relationships.

\begin{proposition}\label{prop:RP-group4}
    Ranked Pairs
    satisfies the majority criterion and fails to satisfy Properties~$\alpha$ and Unique-$\alpha$ via a 1-abstention-safe election. 
\end{proposition}
\begin{proof}
    The satisfaction of the majority criterion is immediate as a majority winner is also a Condorcet winner, and thus is always a unique winner under Ranked Pairs. Moreover, since Ranked Pairs is weakly resolute, it suffices to show that Ranked Pairs fails to satisfy Property Unique-$\alpha$ via a 1-abstention-safe election.

    We use the same election $(C, V)$ as in the proof of Proposition~\ref{prop:schulze-group4}. Under Ranked Pairs, the $a$ is declared a winner of $(C, V)$, but in the election $(\{a,b\}, V)$, $b$ is a unique winner. It is easy thus to verify that Ranked Pairs fails to satisfy Property Unique-$\alpha$ via a 1-abstention-safe election.
    \qed
\end{proof}

We now focus on Group~5 relationships.
Readers may notice that the definition of a winner under RP explicitly states that such a winner is a unique winner. Therefore, RP is a weakly resolute voting rule (recall that a weakly resolute voting rule is one that never outputs more than one winner)~\cite{tid:j:clones-dodgson}. This allows us to derive a vast number of collapses under RP, and those collapses also hold under all weakly resolute voting rules.

\begin{theorem}\label{thm:rp-te-tp-uw-nuw-collapse}
    Under the Ranked Pairs voting rule (RP), it holds that
    \begin{enumerate}
        \item each control type in the UW model collapses to the corresponding control type in the NUW model, and 
        \item each partition-based control type using the TP tie-handling rule collapses to the corresponding control type using the TE tie-handling rule, 
        \item all destructive control types that are about partitioning candidates collapse with each other.
    \end{enumerate}
\end{theorem}
\begin{proof}
    The first two cases of the theorem are straightforward because RP is weakly resolute, and so there can never be tied winners in any elections under RP\@. To see why the third case is true, first observe that by the first two cases of this theorem, $\dc\dash\pc\dash\te\dash\uw=\dc\dash\pc\dash\tp\dash\uw=\dc\dash\pc\dash\te\dash\nuw=\dc\dash\pc\dash\tp\dash\nuw$ and the analogous collapses with respect to RPC also hold, thus concluding this proof. Finally, recall first that under RP (and in fact, under every voting rule~\cite{hem-hem-men:j:search-versus-decision}), DC-PC-TE-UW and DC-RPC-TE-UW collapse.\qed
\end{proof}

We prove that beyond the universal containments and the collapses mentioned above, no other containments/collapses hold under Ranked Pairs\@. Our separations are listed in Table~\ref{tab:ranked-pairs-relationships}.

\subsection{Bucklin}\label{subsec:bucklin}

Let $E=(C,V)$ be an election, where $C$ is a finite set of candidates, $V$ is a collection of voters, and each voter submits a linear order on $C$. 
The Bucklin voting rule proceeds in stages $i \in \{1, \ldots, \|C\|\}$ as follows.
During a stage $i$, count the number of times each candidate $c$ appears in the top $i$ positions of a vote and denote that number by $b_{i}(c)$, and let $C_i' = \{c \in C \mid b_i(c) \geq \lfloor \|V\|/2 \rfloor + 1\}$. If $C_i'$ is not empty, then all the candidates in $C_i'$ are declare winners. Otherwise, proceed to stage $i+1$.

Notice that in the final stage $\|C\|$, i.e., when $i = \|C\|$, for each $c \in C$, $b^m(c) = \|C\|$ and so $C_{\|C\|}'=C$.
Therefore when $C\neq\emptyset$, Bucklin always elects at least one winner.

\begin{example}\label{ex:bucjlin-basic}
Consider the following election, where $C=\{a,b,c,d\}$ and $V$ contains the following five votes:
\[ 
\begin{array}{@{}l@{}}	
	v_1: a > d > b > c\\	
	v_2: a > d > b > c\\	
	v_3: b > d > a > c\\	
	v_4: b > d > a > c\\	
	v_5: c > d > a > b	
\end{array}
\] 
The strict majority threshold---that is, $\lfloor \|V\|/2\rfloor +1$---is $3$. At stage~1, the first-place count is $a:2$, $b:2$, $c:1$, and $d:0$, therefore no candidate reaches a strict majority. At stage~2, candidate $d$ appears in the top two positions on every vote, thus earning a score of  $5\ge 3$. No other candidate crosses that threshold in stage~2, so $d$ is a unique winner. 
\end{example}

We first address the Group~4 relationships.

\begin{proposition}\label{prop:bucklin-group4}
    Bucklin
    satisfies the majority criterion and fails to satisfy Properties~$\alpha$ and Unique-$\alpha$ via a 1-abstention-safe election. 
\end{proposition}
\begin{proof}
    It is immediate from the definition of Bucklin that it satisfies the majority criterion.

    Let $C=\{a,b,c\}$ and let $V$ contain the following 7 votes:
		\[
		\begin{array}{l}
			v_1: a > b > c,\\
			v_2: c > a > b,\\
			v_3: b > a > c,\\
			v_4: c > a > b,\\
			v_5: c > a > b,\\
            v_6: a > b > c,\\
            v_7: b > c > a.
		\end{array}
        \]
		The majority threshold is $4$. During stage 2, $a$ is a unique winner with score 6, while $b$ has score 4 and $c$ has score 4. However, upon deleting candidate $b$, candidate $c$ becomes a unique majority winner of the election $(\{a,c\}, V)$. Thus Bucklin fails to satisfy both Properties~$\alpha$ and Unique-$\alpha$. Moreover, deleting a single vote makes no candidate a stage~1 winner, and cannot prevent $a$ from being a unique winner, so $(C, V)$ is 1-abstention safe.
    \qed
\end{proof}

For Group~5, under Bucklin, we prove that beyond the universal containments/collapses, no other containments/collapses hold. Our separation witnesses are listed in Table~\ref{tab:bucklin-relationships}.

\subsection{Fallback}\label{subsec:fallback}

Fallback voting is a ``hybrid'' of approval and Bucklin. 
Each voter approves of a subset of candidates, and ranks their approved candidates.
All unapproved candidates are left unranked. 
Bucklin is then applied to the election. However, since not all candidates are ranked, it is possible that during stage $\|C\|$, $C_{\|C\|}' = \emptyset$. If that happens, then approval voting is conducted using the approvals, thereby providing a ``fallback'' option.

\begin{example}
Consider the following example, which is based on Example~\ref{ex:bucjlin-basic}.

\[ \begin{array}{@{}l@{}}	
	v_1: a\\	
	v_2: a > d\\	
	v_3: b > d\\	
	v_4: b\\	
	v_5: c	
\end{array}
\]

In stage~$1$, the first-place count is $b_1(a)=2$, $b_1(b)=2$, $b_1(c)=1$, and $b_1(d)=0$, therefore no candidate reaches a strict majority. In stage~2, only $v_2$ and $v_3$ approved a second-place candidate, so $b_2(a)=2$, $b_2(b)=2$, $b_2(c)=1$, and $b_2(d)=2$. The scores do not change in future stages, so no candidate reaches the majority threshold at any stage. Therefore,  Fallback will revert to approval and select all approval winners, i.e., $a$, $b$, and $d$.
\end{example} 

Bucklin can be seen as a special case of Fallback: if every voter approves all candidates, and the order of approval is consistent with the original linear order, then the Fallback score at each level is exactly the same as Bucklin's, and therefore the winners are exactly the same. Thus, every Bucklin separation witness in earlier theorems and in Table~\ref{tab:bucklin-relationships} can be directly converted into a Fallback separation witness: simply by getting all voters to approve all candidates. Therefore, under Fallback voting, there are no collapses/containments beyond the universal ones.

}

{

\providecommand{\Imm}{IMM}
\providecommand{\HHMColl}{HHM}
\providecommand{\SCTCont}{SCT}
\providecommand{\CopDCColl}{Cor.~\ref{cor:copeland-dc}}

\providecommand{\LlullDCColl}{Cor.~\ref{cor:copeland-dc}}
\providecommand{\LlullRPCteUW}{Thm.~\ref{thm:llull-cc-rpc-te-uw-subset-pc-te-uw}}
\providecommand{\LlullRPCteNUW}{Cor.~\ref{cor:llull-cc-rpc-te-nuw-subset-pc-te-nuw}}
\providecommand{\LlullRPCtp}{Thm.~\ref{thm:llull-cc-rpc-tp-nuw-subset-pc-te-nuw}}
\providecommand{\LlullPCtpUW}{Thm.~\ref{thm:llull-cc-pc-tp-uw-subset-pc-te-nuw}}

\providecommand{\SchulzeDCColl}{Thm.~\ref{thm:dc-te-and-tp-uw-collapse-schulze}}

\providecommand{\RPRes}{Thm.~\ref{thm:rp-te-tp-uw-nuw-collapse}}

\newcommand{\currentPath}{appendix/}

\mainUnit{Separation Witnesses and Tables}
The following appendix tables use a compressed witness format. 
In each witness, $m$ denotes the number of candidates, and $V$ denotes the multiset of voters. 
For score-based rules, each string in $V$ denotes one score ballot over the candidates in their fixed order. 
For ranking-based rules, each string denotes one linear order over the candidates. 
Search metadata such as seeds, runtimes, iteration counts, and verification timestamps are omitted.

In the relationship tables below, ``\Imm{}'' denotes an immediate containment following from the UW/NUW winner-model definitions. 
The symbol ``\HHMColl{}'' denotes the seven universal collapses recalled in Section~\ref{sec:results}: the six pairwise collapses inside the destructive TE candidate-partition class and the collapse between the two destructive TP-NUW candidate-partition variants. 
The symbol ``\SCTCont{}'' denotes the universal containments from SCT recalled in Section~\ref{sec:results}. 
Rule-specific theorem or corollary references indicate relationships proved in the corresponding rule section.

Witness identifiers such as 2NRV.1, refer to the corresponding witness catalogues in this appendix. 
For rows classified as INCOMP, two witness identifiers separated by ``+'' refer to the two noncontainment directions. 
If the same witness proves both directions, we write the witness identifier once followed by ``(both)''. 
Rows marked with ``?'' are not fully classified in this work and are thus open; the listed witness, if any, records the known separation direction, while the remaining direction is open as discussed in the corresponding rule section.

Notice that all our tables only contain results that are about the compatibility group~5, and that is because our results from Section~\ref{sec:nonpartition} give us precisely those results we seek and they are recalled at the relevant subsections in Section~\ref{sec:results}. However, although Theorem~\ref{t:all-cc-dc} gives us incomparability results for $12\times 12 = 144$ pairs of control types in Group~5, our tables includes all $\binom{24}{2} = 276$ possible pairs within Group~5 for the sake of completeness.

\secondaryUnit{$k$-NRV Witness Tables}\label{sec:knrv-witness}

By Theorem~\ref{thm:knrv-invariance} and Corollary~\ref{cor:knrv-separation-transfer}, every separation witness for $2$-NRV also yields a separation witness for $k$-NRV for every $k>2$. 
Therefore, it is enough to list the witnesses for $2$-NRV in this appendix. 
The corresponding witnesses for larger values of $k$ are obtained by the score-scaling transfer described in Section~\ref{subsec:kNRV}.

\begin{longtable}{p{0.15\textwidth}p{0.78\textwidth}}
	\caption{Distinct separation witnesses for 2NRV.}
	\label{tab:2nrv-witness-catalogue}\\
	\hline
	ID & Witness election \\
	\hline
	\endfirsthead
	\hline
	ID & Witness election \\
	\hline
	\endhead
	2NRV.1 & $m=3,\; V=\{010,011,101\}$ \\
	\hline
	2NRV.2 & $m=4,\; V=\{1012,1101,1120\}$ \\
	\hline
	2NRV.3 & $m=3,\; V=\{020,202,212\}$ \\
	\hline
	2NRV.4 & $m=3,\; V=\{122,122,201\}$ \\
	\hline
	2NRV.5 & $m=4,\; V=\{0111,0210,2010\}$ \\
	\hline
	2NRV.6 & $m=3,\; V=\{001,210,211\}$ \\
	\hline
	2NRV.7 & $m=3,\; V=\{002,120,221\}$ \\
	\hline
	2NRV.8 & $m=4,\; V=\{0210,1220,2010\}$ \\
\hline
\end{longtable}

\begin{longtable}{p{0.27\textwidth}p{0.27\textwidth}p{0.19\textwidth}p{0.20\textwidth}}
	\caption{Separations and collapses for $2$-NRV.}
	\label{tab:2nrv-relationships}\\
	\hline
	$\mathcal{T}$ & $\mathcal{T}'$ & Classification & Justification(s) \\
	\hline
	\endfirsthead
	\hline
	$\mathcal{T}$ & $\mathcal{T}'$ & Classification & Justification(s) \\
	\hline
	\endhead
	$\text{CC-PV-TE-UW}$ & $\text{CC-PV-TE-NUW}$ & $\subsetneq$ & \Imm{} + 2NRV.1 \\
	\hline
	$\text{CC-PV-TE-UW}$ & $\text{CC-PV-TP-UW}$ & INCOMP & 2NRV.1 + 2NRV.2 \\
	\hline
	$\text{CC-PV-TE-UW}$ & $\text{CC-PV-TP-NUW}$ & INCOMP & 2NRV.3 + 2NRV.1 \\
	\hline
	$\text{CC-PV-TE-UW}$ & $\text{CC-RPC-TE-UW}$ & INCOMP & 2NRV.1 (both)\\
	\hline
	$\text{CC-PV-TE-UW}$ & $\text{CC-RPC-TE-NUW}$ & INCOMP & 2NRV.4 + 2NRV.1 \\
	\hline
	$\text{CC-PV-TE-UW}$ & $\text{CC-RPC-TP-UW}$ & INCOMP & 2NRV.1 + 2NRV.5 \\
	\hline
	$\text{CC-PV-TE-UW}$ & $\text{CC-RPC-TP-NUW}$ & INCOMP & 2NRV.3 + 2NRV.1 \\
	\hline
	$\text{CC-PV-TE-UW}$ & $\text{CC-PC-TE-UW}$ & INCOMP & 2NRV.1 (both)\\
\hline
\end{longtable}

\secondaryUnit{Copeland Witness Tables}\label{app:copeland-witnesses}

In Table~\ref{tab:copeland-relationships}, ''\CopDCColl{}'' denotes the Copeland$^\alpha$ destructive candidate-partition collapse for $\alpha\in[0,1)$ proved in Section~\ref{subsec:copeland-05}.

\begin{longtable}{p{0.22\textwidth}p{0.82\textwidth}}
\caption{Catalogue of distinct separation witnesses for Copeland.}
\label{tab:copeland-witness-catalogue}\\
\hline
ID & Witness election \\
\hline
\endfirsthead
\hline
ID & Witness election \\
\hline
\endhead
Copeland.1 & $m=3,\; V=\{abc,abc,bac,bac\}$ \\
\hline
Copeland.2 & $m=3,\; V=\{abc,bac,cab\}$ \\
\hline
Copeland.3 & $m=4,\; V=\{acbd,cbda,dbac\}$ \\
\hline
Copeland.4 & $m=2,\; V=\{ab,ba\}$ \\
\hline
Copeland.5 & $m=3,\; V=\{abc,bac\}$ \\
\hline
Copeland.6 & $m=4,\; V=\{bacd,cadb,dbac\}$ \\
\hline
Copeland.7 & $m=4,\; V=\{abcd,cdab\}$ \\
\hline
Copeland.8 & $m=2,\; V=\{ab\}$ \\
\hline
\end{longtable}

\begin{longtable}{p{0.27\textwidth}p{0.27\textwidth}p{0.19\textwidth}p{0.20\textwidth}}
	\caption{Separations and collapses for Copeland.}
	\label{tab:copeland-relationships}\\
	\hline
	$\mathcal{T}$ & $\mathcal{T}'$ & Classification & Justification(s) \\
	\hline
	\endfirsthead
	\hline
	$\mathcal{T}$ & $\mathcal{T}'$ & Classification & Justification(s) \\
	\hline
	\endhead
	$\text{CC-PV-TE-UW}$ & $\text{CC-PV-TE-NUW}$ & $\subsetneq$ & \Imm{} + Copeland.1 \\
	\hline
	$\text{CC-PV-TE-UW}$ & $\text{CC-PV-TP-UW}$ & INCOMP & Copeland.2 + Copeland.3 \\
	\hline
	$\text{CC-PV-TE-UW}$ & $\text{CC-PV-TP-NUW}$ & INCOMP & Copeland.2 + Copeland.4 \\
	\hline
	$\text{CC-PV-TE-UW}$ & $\text{CC-RPC-TE-UW}$ & INCOMP & Copeland.2 + Copeland.5 \\
	\hline
	$\text{CC-PV-TE-UW}$ & $\text{CC-RPC-TE-NUW}$ & INCOMP & Copeland.2 + Copeland.4 \\
	\hline
	$\text{CC-PV-TE-UW}$ & $\text{CC-RPC-TP-UW}$ & INCOMP & Copeland.2 + Copeland.6 \\
	\hline
	$\text{CC-PV-TE-UW}$ & $\text{CC-RPC-TP-NUW}$ & INCOMP & Copeland.2 + Copeland.4 \\
	\hline
	$\text{CC-PV-TE-UW}$ & $\text{CC-PC-TE-UW}$ & INCOMP & Copeland.2 + Copeland.5 \\
\hline
\end{longtable}

\secondaryUnit{Llull Witness Tables}\label{app:llull-witnesses}
In Table~\ref{tab:llull-relationships}, the symbols ``\LlullDCColl{}'', ``\LlullRPCteUW{}'', ``\LlullRPCteNUW{}'', ``\LlullRPCtp{}'', and ``\LlullPCtpUW{}'' refer to the Llull-specific collapse and containment results proved in Section~\ref{subsec:llull}.

\begin{longtable}{p{0.15\textwidth}p{0.78\textwidth}}
\caption{Distinct separation witnesses for Llull.}
\label{tab:llull-witness-catalogue}\\
\hline
ID & Witness election \\
\hline
\endfirsthead
\hline
ID & Witness election \\
\hline
\endhead
Llull.1 & $m=3,\; V=\{abc,abc,bac,bac\}$ \\
\hline
Llull.2 & $m=3,\; V=\{abc,bac,cab\}$ \\
\hline
Llull.3 & $m=5,\; V=\{abecd,cdbea,edbac\}$ \\
\hline
Llull.4 & $m=4,\; V=\{acbd,bcda,dacb\}$ \\
\hline
Llull.5 & $m=3,\; V=\{abc,abc,abc\}$ \\
\hline
Llull.6 & $m=3,\; V=\{abc,abc,cba,cba\}$ \\
\hline
Llull.7 & $m=3,\; V=\{abc,bca,cab\}$ \\
\hline
Llull.8 & $m=3,\; V=\{abc,abc,bac,cba\}$ \\
\hline
\end{longtable}

\begin{longtable}{p{0.27\textwidth}p{0.27\textwidth}p{0.19\textwidth}p{0.23\textwidth}}
	\caption{Separations and collapses for Llull.}
	\label{tab:llull-relationships}\\
	\hline
	$\mathcal{T}$ & $\mathcal{T}'$ & Classification & Justification(s) \\
	\hline
	\endfirsthead
	\hline
	$\mathcal{T}$ & $\mathcal{T}'$ & Classification & Justification(s) \\
	\hline
	\endhead
	$\text{CC-PV-TE-UW}$ & $\text{CC-PV-TE-NUW}$ & $\subsetneq$ & \Imm{} + Llull.1 \\
	\hline
	$\text{CC-PV-TE-UW}$ & $\text{CC-PV-TP-UW}$ & INCOMP & Llull.2 + Llull.3 \\
	\hline
	$\text{CC-PV-TE-UW}$ & $\text{CC-PV-TP-NUW}$ & INCOMP & Llull.2 + Llull.1 \\
	\hline
	$\text{CC-PV-TE-UW}$ & $\text{CC-RPC-TE-UW}$ & INCOMP & Llull.2 + Llull.1 \\
	\hline
	$\text{CC-PV-TE-UW}$ & $\text{CC-RPC-TE-NUW}$ & INCOMP & Llull.2 + Llull.1 \\
	\hline
	$\text{CC-PV-TE-UW}$ & $\text{CC-RPC-TP-UW}$ & INCOMP & Llull.2 + Llull.4 \\
	\hline
	$\text{CC-PV-TE-UW}$ & $\text{CC-RPC-TP-NUW}$ & INCOMP & Llull.2 + Llull.1 \\
	\hline
	$\text{CC-PV-TE-UW}$ & $\text{CC-PC-TE-UW}$ & INCOMP & Llull.2 + Llull.1 \\
\hline
\end{longtable}

\secondaryUnit{Schulze Witness Tables}\label{app:schulze-witnesses}
In Table~\ref{tab:schulze-relationships}, ``\SchulzeDCColl{}'' denotes the Schulze destructive candidate-partition collapse proved in Section~\ref{subsec:schulze}.

\begin{longtable}{p{0.12\textwidth}p{0.82\textwidth}}
\caption{Catalogue of distinct separation witnesses for Schulze.}
\label{tab:schulze-witness-catalogue}\\
\hline
ID & Witness election \\
\hline
\endfirsthead
\hline
ID & Witness election \\
\hline
\endhead
Schulze.1 & $m=3,\; V=\{abc,abc,bac,bac\}$ \\
\hline
Schulze.2 & $m=3,\; V=\{abc,bac,cab\}$ \\
\hline
Schulze.3 & $m=5,\; V=\{abdec,cebad,debac\}$ \\
\hline
Schulze.4 & $m=4,\; V=\{acbd,cbda,dbac\}$ \\
\hline
Schulze.5 & $m=3,\; V=\{abc,abc,abc\}$ \\
\hline
Schulze.6 & $m=3,\; V=\{abc,abc,cba,cba\}$ \\
\hline
Schulze.7 & $m=3,\; V=\{abc,bca,cab\}$ \\
\hline
Schulze.8 & $m=3,\; V=\{abc,abc,bac,cba\}$ \\
\hline
\end{longtable}

\begin{longtable}{p{0.27\textwidth}p{0.27\textwidth}p{0.19\textwidth}p{0.20\textwidth}}
\caption{Separations and collapses for Schulze.}
\label{tab:schulze-relationships}\\
\hline
$\mathcal{T}$ & $\mathcal{T}'$ & Classification & Justification(s) \\
\hline
\endfirsthead
\hline
$\mathcal{T}$ & $\mathcal{T}'$ & Classification & Justification(s) \\
\hline
\endhead
$\text{CC-PV-TE-UW}$ & $\text{CC-PV-TE-NUW}$ & $\subsetneq$ & \Imm{} + Schulze.1 \\
\hline
$\text{CC-PV-TE-UW}$ & $\text{CC-PV-TP-UW}$ & INCOMP & Schulze.2 + Schulze.3 \\
\hline
$\text{CC-PV-TE-UW}$ & $\text{CC-PV-TP-NUW}$ & INCOMP & Schulze.2 + Schulze.1 \\
\hline
$\text{CC-PV-TE-UW}$ & $\text{CC-RPC-TE-UW}$ & INCOMP & Schulze.2 + Schulze.1 \\
\hline
$\text{CC-PV-TE-UW}$ & $\text{CC-RPC-TE-NUW}$ & INCOMP & Schulze.2 + Schulze.1 \\
\hline
$\text{CC-PV-TE-UW}$ & $\text{CC-RPC-TP-UW}$ & INCOMP & Schulze.2 + Schulze.4 \\
\hline
$\text{CC-PV-TE-UW}$ & $\text{CC-RPC-TP-NUW}$ & INCOMP & Schulze.2 + Schulze.1 \\
\hline
$\text{CC-PV-TE-UW}$ & $\text{CC-PC-TE-UW}$ & INCOMP & Schulze.2 + Schulze.1 \\
\hline
\end{longtable}

\secondaryUnit{Ranked Pairs Witness Tables}\label{app:rp-witnesses}
In Table~\ref{tab:ranked-pairs-relationships}, ``\RPRes{}'' denotes Theorem~\ref{thm:rp-te-tp-uw-nuw-collapse}, the resoluteness-based collapse for Ranked Pairs that was proven in Section~\ref{subsec:ranked-pairs}.

\begin{longtable}{p{0.15\textwidth}p{0.78\textwidth}}
\caption{Distinct separation witnesses for Ranked Pairs.}
\label{tab:ranked-pairs-witness-catalogue}\\
\hline
ID & Witness election \\
\hline
\endfirsthead
\hline
ID & Witness election \\
\hline
\endhead
RP.1 & $m=3,\; V=\{abc,bac,cab\}$ \\
\hline
RP.2 & $m=3,\; V=\{abc,bca,cab\}$ \\
\hline
RP.3 & $m=3,\; V=\{abc,abc,abc\}$ \\
\hline
RP.4 & $m=5,\; V=\{beacd,cdeab,acbed,aecbd,cbdae,adbec,abedc,cdabe,$\\
&$bdcae,bdeac\}$ \\
\hline
RP.5 & $m=4,\; V=\{dacb,cbda,abcd\}$ \\
\hline
RP.6 & $m=5,\; V=\{aebdc,decba,cebad,bdcea,ceabd,adebc,decab,abdce,ebadc\}$ \\
\hline
RP.7 & $m=5,\; V=\{dbaec,deacb,abced,cbdae,dabec,dcbea,aedcb,ebadc,$\\
&$baedc,cedab\}$ \\
\hline
RP.8 & $m=5,\; V=\{bedca,eacdb,abdce,cdbae,abdce,acbde,cdeab,becad,$\\
&$cebad,aedbc,bedac,dacbe\}$ \\
\hline
RP.9 & $m=4,\; V=\{badc,adcb,bdca,acbd\}$ \\
\hline
\end{longtable}

\begin{longtable}{p{0.27\textwidth}p{0.27\textwidth}p{0.19\textwidth}p{0.20\textwidth}}
\caption{Separations and collapses for Ranked Pairs.}
\label{tab:ranked-pairs-relationships}\\
\hline
$\mathcal{T}$ & $\mathcal{T}'$ & Classification & Justification(s) \\
\hline
\endfirsthead
\hline
$\mathcal{T}$ & $\mathcal{T}'$ & Classification & Justification(s) \\
\hline
\endhead
$\text{CC-PV-TE-UW}$ & $\text{CC-PV-TE-NUW}$ & EQ & \Imm{} + \RPRes{} \\
\hline
$\text{CC-PV-TE-UW}$ & $\text{CC-PV-TP-UW}$ & EQ & \RPRes{} \\
\hline
$\text{CC-PV-TE-UW}$ & $\text{CC-PV-TP-NUW}$ & EQ & \RPRes{} \\
\hline
$\text{CC-PV-TE-UW}$ & $\text{CC-RPC-TE-UW}$ & INCOMP & RP.1 + RP.2 \\
\hline
$\text{CC-PV-TE-UW}$ & $\text{CC-RPC-TE-NUW}$ & INCOMP & RP.1 + RP.2 \\
\hline
$\text{CC-PV-TE-UW}$ & $\text{CC-RPC-TP-UW}$ & INCOMP & RP.1 + RP.2 \\
\hline
$\text{CC-PV-TE-UW}$ & $\text{CC-RPC-TP-NUW}$ & INCOMP & RP.1 + RP.2 \\
\hline
$\text{CC-PV-TE-UW}$ & $\text{CC-PC-TE-UW}$ & INCOMP & RP.1 + RP.2 \\
\hline
\end{longtable}

\secondaryUnit{Bucklin Witness Tables}\label{app:bucklin-witnesses}

\begin{longtable}{p{0.12\textwidth}p{0.82\textwidth}}
\caption{Catalogue of distinct separation witnesses for Bucklin.}
\label{tab:bucklin-witness-catalogue}\\
\hline
ID & Witness election \\
\hline
\endfirsthead
\hline
ID & Witness election \\
\hline
\endhead
Bucklin.1 & $m=3,\; V=\{abc,abc,bac,bac\}$ \\
\hline
Bucklin.2 & $m=3,\; V=\{abc,bac,cab\}$ \\
\hline
Bucklin.3 & $m=4,\; V=\{badc,cbda,dcab\}$ \\
\hline
Bucklin.4 & $m=4,\; V=\{adcb,bdca,cbda\}$ \\
\hline
Bucklin.5 & $m=3,\; V=\{abc,abc,abc\}$ \\
\hline
Bucklin.6 & $m=3,\; V=\{abc,abc,bca,cab\}$ \\
\hline
Bucklin.7 & $m=3,\; V=\{abc,abc,bac,cba\}$ \\
\hline
Bucklin.8 & $m=4,\; V=\{adbc,badc,cdab\}$ \\
\hline
\end{longtable}

\begin{longtable}{p{0.27\textwidth}p{0.27\textwidth}p{0.19\textwidth}p{0.23\textwidth}}
\caption{Separations and collapses for Bucklin.}
\label{tab:bucklin-relationships}\\
\hline
$\mathcal{T}$ & $\mathcal{T}'$ & Classification & Justification(s) \\
\hline
\endfirsthead
\hline
$\mathcal{T}$ & $\mathcal{T}'$ & Classification & Justification(s) \\
\hline
\endhead
$\text{CC-PV-TE-UW}$ & $\text{CC-PV-TE-NUW}$ & $\subsetneq$ & \Imm{} + Bucklin.1 \\
\hline
$\text{CC-PV-TE-UW}$ & $\text{CC-PV-TP-UW}$ & INCOMP & Bucklin.2 + Bucklin.3 \\
\hline
$\text{CC-PV-TE-UW}$ & $\text{CC-PV-TP-NUW}$ & INCOMP & Bucklin.2 + Bucklin.1 \\
\hline
$\text{CC-PV-TE-UW}$ & $\text{CC-RPC-TE-UW}$ & INCOMP & Bucklin.2 + Bucklin.1 \\
\hline
$\text{CC-PV-TE-UW}$ & $\text{CC-RPC-TE-NUW}$ & INCOMP & Bucklin.2 + Bucklin.1 \\
\hline
$\text{CC-PV-TE-UW}$ & $\text{CC-RPC-TP-UW}$ & INCOMP & Bucklin.2 + Bucklin.4 \\
\hline
$\text{CC-PV-TE-UW}$ & $\text{CC-RPC-TP-NUW}$ & INCOMP & Bucklin.2 + Bucklin.1 \\
\hline
$\text{CC-PV-TE-UW}$ & $\text{CC-PC-TE-UW}$ & INCOMP & Bucklin.2 + Bucklin.1 \\
\hline
\end{longtable}

\secondaryUnit{Fallback Witness Tables}\label{app:fallback-witnesses}
No separate fallback witness table is included: by approving all candidates in every vote, each Bucklin witness can be interpreted as a fallback witness with the same winner sets.

}
\end{document}